%% file: main.tex
\documentclass[letterpaper, 10 pt, conference]{ieeeconf}
\usepackage{cite}
\usepackage{amsmath,amssymb,amsfonts}
\usepackage{graphicx}
\usepackage[hidelinks]{hyperref}
\hypersetup{hidelinks=true}
\usepackage{textcomp}
\def\BibTeX{{\rm B\kern-.05em{\sc i\kern-.025em b}\kern-.08em
    T\kern-.1667em\lower.7ex\hbox{E}\kern-.125emX}}

\input{preamble/common}
\input{preamble/reviewDraftClean} % <--- If you comment the previous line and have errors, uncomment this line for one compilation (you can remove it again afterwards)

\title{Constrained Koopman Control Parametrization: Convex Controller Synthesis for Tracking under Input Constraints}

\author{Taha Ondogan,
Ran Jing, Andrew P. Sabelhaus, and Roberto Tron
\thanks{T. Ondogan, R. Jing, A.P. Sabelhaus, and R. Tron are with the Department of Mechanical Engineering, Boston University, Boston, MA 02215, USA (e-mail: {tondogan, rjing, asabelha, tron}@bu.edu).}
\thanks{A. Sabelhaus and R. Tron are also with the Division of Systems Engineering, Boston University, Boston, MA 02215, USA.}
\thanks{This work was supported by the U.S. National Science Foundation under Award 2432394.}
}

\NewDocumentCommand\kocop{}{KoCoP}
\WithSuffix\NewDocumentCommand\kocop*{}{Koopman Control Parametrization}
\NewDocumentCommand\ckocop{}{cKoCoP}
\WithSuffix\NewDocumentCommand\ckocop*{}{constrained Koopman Control Parametrization}

\begin{document}
\IEEEoverridecommandlockouts
\maketitle
\thispagestyle{empty}
%\pagestyle{empty}

%%%%%%%%%%%%%%%%%%%%%%%%%%%%%%%%%%%%%%%%%%%%%%%%%%%%%%%%%%%%%%%%%%%%%%%%%%%%%%%%
\begin{abstract}
Koopman operator theory represents nonlinear systems with linear and bilinear terms in the state and input observables, greatly simplifying control design. However, incorporating input constraints into Koopman-based controller synthesis remains challenging, limiting their applicability to constrained control tasks such as tracking. This paper develops a constrained Koopman Control Parametrization (cKoCoP) framework for synthesizing state-feedback controllers that satisfy prescribed input constraints with Lyapunov-certified stability of the identified lifted closed loop. The cKoCoP offers nonlinear state-feedback via linear-in-observables feedback, enabling linear controller design. The proposed controller is obtained through convex optimization: we incorporate ellipsoidal or polytopic input constraints with the S-procedure or Farkas-lemma-based conditions. We also introduce a Control Lyapunov Function (CLF) for trajectory tracking which incorporates input constraints. The approach is validated on two canonical nonlinear systems, a double pendulum and a unicycle, demonstrating data-driven tracking while respecting input constraints. Experiments validate the method's robustness to measurement noise, and runtime improvements over Koopman Model Predictive Control (KMPC) baselines. The results demonstrate that cKoCoP achieves competitive tracking performance with substantially reduced online computational cost, supporting its potential for real-time constrained tracking applications.
\end{abstract}

\section{Introduction}
Nonlinear dynamics arises in a wide range of robotic systems. However, nonlinear control and stability analyses are limited compared to linear systems, as linear systems offer simple-to-formulate closed form controllers (e.g., LQR \cite{kalman1960general}). 
% Therefore many control algorithms employ linearization before design. However these linearizations are often local around the equilibrium point. 
To overcome limitations of locality and linearization while retaining the benefits of linear methods, Koopman operators transform a nonlinear system into a higher (potentially infinite) dimensional linear system. Koopman operator control methods have been adopted in robotics systems \cite{shi2026koopman, khorshidi2024centroidal}, fluid systems \cite{arbabi2017ergodic, schmid2022dynamic}, and power systems \cite{matavalam_datadriven_2024}. However finding such high dimensional systems is often challenging, as the model is often not available, and we can only get approximations from data that often contain noise. Furthermore, it results in finite-dimensional approximation of the Koopman operator. Therefore, many algorithms utilize learning these systems from the data either with a preselected set of observables \cite{guo2021koopman, abraham2017model}, or deep learned observables
\cite{lusch2018deep, fathi2023course, mondal2024efficient}. Extended Dynamic Mode Decomposition (EDMD) algorithm is widely used to fit a linear operator between a dictionary of nonlinear observable functions via least squares from the data \cite{williams2015data,korda2018convergence,Brunton2022ModernKoopman}.

Koopman operators have been shown to capture long-term system evolution effectively \cite{azencot2020forecasting}, making them well suited for tracking applications. In this context, receding-horizon methods such as model predictive control (MPC) are widely adopted in the Koopman literature \cite{abraham2017model, joglekar2023analytical, korda2018linear}. These approaches typically follow a two-stage framework: first, a data-driven linear Koopman model is identified using a set of predefined observables (such as EDMD); second, an MPC controller is designed and applied to the resulting linear controlled Koopman system approximation. 
In the standard EDMD-based setting, the resulting controller is often a linear MPC problem in the lifted coordinates. To improve modeling capability, \cite{bruder2021advantages} introduced bilinear Koopman system representations within the MPC framework, enabling a broader class of input-affine nonlinear systems to be captured. This extension has demonstrated improved tracking performance in several settings, including nonlinear MPC (NMPC) and bilinear MPC (BMPC). Despite these advantages, such approaches remain computationally demanding, as they require solving an optimization problem over a finite prediction horizon at every control step.

A crucial property desired in tracking control systems is the satisfaction of actuation constraints, especially since data-driven Koopman approximations may exhibit model mismatch relative to the actual nonlinear dynamics. While MPC naturally allows the incorporation of input constraints, it further increases the online computational complexity \cite{rawlings2020model}. Moreover, even if the admissible set is convex in the original state space, its image under the nonlinear lifting map may become nonconvex or difficult to characterize explicitly, further complicating the use of MPC in the lifted Koopman coordinates.

A preliminary version of this work proposed the Koopman Control Parametrization
(\kocop{}), a lifted-state feedback controller structure \cite{ondogan2025koopman}.
The method identifies a linear closed-loop structure in the lifted coordinates and
synthesizes the controller through a convex optimization problem that is solved once
offline. As a static state-feedback controller, \kocop{} is computationally inexpensive
compared with MPC formulations. However, the original \kocop{} formulation does not
account for actuation limits and is restricted to stabilization about the origin.

Building on the preliminaries reviewed in \cref{sec:preliminaries} and the original
\kocop{} construction summarized in \cref{sec:kocop}, this paper makes the following contributions:
\begin{itemize}
    \item An Alternate Convex Search procedure is introduced to the previous version of \kocop{} algorithm that addresses the biconvex nature of the problem
    \item A constrained \kocop{} formulation
    that guarantees satisfaction of actuation constraints over a predefined set of
    lifted states in \cref{sec:constraints} is developed. This is achieved through convex reformulations based on
    \emph{Farkas' Lemma} for polytopic constraints and the \emph{S-procedure} for
    ellipsoidal constraints, thereby preserving the convexity of the resulting
    optimization problem for a fixed Lyapunov matrix.
    \item We introduce a feedforward term that enables
    stabilization about nonzero equilibrium points by exploiting the corresponding
    error dynamics in \cref{sec:setpoint}.
    \item We develop a control-Lyapunov-function
    (CLF)-based tracking formulation that dynamically updates the feedforward term
    to track time-varying references while satisfying the prescribed input
    constraints in \cref{sec:clf_tracking}.
\end{itemize}

In the proposed approach, a stabilizing controller is first synthesized offline using
the \ckocop*{} formulation. Then, an online CLF-based quadratic program is
used to compute the tracking input under either polytopic or ellipsoidal actuation
constraints. The experimental setup is described in \cref{sec:experiments}, and presented tracking scenarios with actuation limits, robustness
results, and performance and runtime comparisons with several MPC baselines. 

\section{Notation and Preliminaries}\label{sec:preliminaries}

We use $t$ to denote the discrete-time index along a trajectory and $n$ to denote the sample 
index in a dataset. For a vector-valued function 
$\vpsi_x:\mathbb{R}^{n_x}\to\mathbb{R}^{n_{\vpsi_x}}$, we write $\vpsi_{x,t}:=\vpsi_x(x_t)$ for its value along a trajectory at time $t$, and 
$\vpsi_x^{(n)}:=\vpsi_x(x^{(n)})$ for its value at the $n$-th data sample. The $i$-th 
component of a vector $v$ is denoted by $v_i$. In particular, $g_i^\top$ denotes the 
$i$-th row of a matrix $G$. We use 
$[M]_{ij}$ to denote the $(i,j)$-th entry of a matrix $M$.

\subsection{Kronecker Product}

\begin{definition}
  Given $A\in\mathbb{R}^{m\times n}$ and $B\in\mathbb{R}^{p\times q}$, the Kronecker product 
$A\kron B\in\mathbb{R}^{mp\times nq}$ is defined as
  \begin{equation}\label{eq:kron}
    A\kron B = \bmat{
      A_{1,1} B & \cdots & A_{1,n}B \\
      \vdots & \ddots &           \vdots \\
      A_{m,1} B & \cdots & A_{m,n} B
    },
  \end{equation}
  \vspace{-1mm}
\end{definition}

The following standard identities follow directly from the definition of the Kronecker product.

\begin{lemma}
The Kronecker product distributes over matrix multiplication according to
\begin{equation}\label{eq:kron mixed}
(AC)\kron(BD)=(A\kron B)(C\kron D).
\end{equation}
\end{lemma}

\begin{corollary}\label{cor:kronecker_vectors}
Let $a\in\mathbb{R}^{n}$, $b\in\mathbb{R}^{m}$, and 
$M\in\mathbb{R}^{r\times m}$. Then
\begin{equation}
    a\kron (Mb) = (I_n\kron M)(a\kron b).
\end{equation}
\end{corollary}

\begin{lemma}
    Given two vectors $a \in \mathbb{R}^n$ and $b \in \mathbb{R}^m$, the following property is satisfied
    \begin{equation}
        a \kron b = (I_n \kron b) a
    \end{equation}
    \label{eq:kron_identity_to_matrix_vector}
\end{lemma}

\begin{lemma}
    Let $a\in\mathbb{R}^{n}$ and $b\in\mathbb{R}^{m}$. There exists a permutation 
    matrix $T_{n,m}\in\mathbb{R}^{nm\times nm}$, called the commutation matrix, such that
    \begin{equation}
        a\kron b = T_{n,m}(b\kron a).
        \label{eq:commutation_matrix}
    \end{equation}
    Its entries are given by
    \begin{equation}
        \begin{aligned}
            [T_{n,m}]_{(k-1)m+\ell,\,(\ell-1)n+k} &= 1,\\
            k &= 1,\ldots,n,\qquad \ell=1,\ldots,m,
        \end{aligned}
    \end{equation}
    and all remaining entries are zero.
\end{lemma}

\subsection{Koopman Operator and Approximation of a Controlled Nonlinear System}

Consider a discrete time autonomous nonlinear dynamical system 
\begin{equation}
    x_{t+1} = F(x_{t}),
    \label{eq:system}
\end{equation}
where $x \in \mathbb{R}^{n_x}$. Let $\vpsi(\cdot) : \mathbb{R}^{n_x} \to \mathbb{R}^{n_{\vpsi_x}}$ be a vector valued function, named \emph{observables}, that maps the state to a higher (possibly infinite) dimensional space.
The Koopman operator $\mathcal{K}$ is a linear operator that acts on observables as

\begin{equation}
    \vpsi(x_{t+1}) = \vpsi(F(x_t)) = [\mathcal{K}\vpsi](x_t).
    \label{eq:koopman}
\end{equation}

The definition of the Koopman operator $\mathcal{K}$ in \eqref{eq:koopman} requires, in general, an infinite number of observables. However, common practical applications rely on a finite dimensional approximation, yielding a Koopman matrix $K$. For instance, the \emph{Extended Dynamic Mode Decomposition} (EDMD) relies on a finite set of functions as a basis for observables, and a least square estimation to obtain a finite dimensional matrix $K \in \mathbb{R}^{n_{\vpsi_x} \times n_{\vpsi_x}}$ as \cite{williams2015data}
\begin{equation}
    \min_K \sum_{i = 1}^{N}  \|\vpsi(x_{t+1}) - K\vpsi(x_t)\|^2_2
\end{equation}

Equation \eqref{eq:koopman} describes a linear evolution of the observables that approximates the dynamics of an autonomous nonlinear system. However, when control inputs are present, it is, in general, not possible to capture the cross-coupled interactions between the system state and the inputs with a purely linear model \cite{peitz2020data}.

In this work, we focus on control-affine nonlinear systems of the form,
\begin{equation}
    x_{t+1} = F(x_{t}) + G(x_{t})u_t.
    \label{eq:ctrl_affine_system}
\end{equation}
To model the interaction between the control input and the state, we consider a bilinear Koopman model which is described in detail in \Cref{sec:kocop}.
%This type of systems admits a bilinear representation that models the interaction between the control input and the state through a Kronecker product operation \cite{bruder2021advantages}. We use this insight to build a convex controller design method in \Cref{sec:kocop}.

\subsection{Lyapunov Stability and Schur Complement}
\begin{lemma}
    The linear discrete time system $x^+ = Ax$ is asymptotically stable if there exists a $ \lambda \in [0, 1)$ and a $P \succ 0$ such that
    \begin{equation}
        A^\top P A - \lambda P \preceq 0
        \label{eq:stability}
    \end{equation}

    Equivalently, \eqref{eq:stability} can be stated using Schur complements as
\begin{equation}
    \begin{bmatrix}
        P & PA \\
        A^\top P & \lambda P
    \end{bmatrix}
    \succeq 0. \label{eq:contraction}
\end{equation}
\end{lemma}

\begin{corollary}
    \label{corollary:positiveInvariance}
    Suppose $P$ is a solution to \eqref{eq:contraction}, then sublevel sets parametrized by level $c$, $\cP_c = \{x \mid x^\top P x \leq c\}$ are positively invariant sets \cite{blanchini1999set}.
\end{corollary}
\begin{proof}
    Let $x \in \cP_c$, so that $V(x) = x^\top P x \leq c$. Then
    \begin{equation}
        V(x^+) = x^\top A^\top P A x \leq \lambda x^\top P x = \lambda V(x) \leq \lambda c.    
    \end{equation}  
    Since $\lambda \in [0,1)$, $x^+ \in \mathcal{P}_c$. The result follows by induction, $\cP_c$ is positively invariant.
\end{proof}

\subsection{$\cS$-procedure and Farkas' Lemma} 
$\cS$-procedure \cite{polik2007survey} and Farkas' lemma \cite{miller2023data} are commonly used convex optimization tools due to their computational efficiency. They are both used to certify a set is contained by another set. While $\cS$-procedure is used for two ellipsoidal sets, Farkas' lemma can be seen as a counterpart for two polytopic sets. 
\begin{lemma}[$\mathcal S$-procedure]\label{lem:s_procedure}
Let
\[
    f_i(x)=x^\top F_i x+2q_i^\top x+c_i,\qquad i=1,2,
\]
where $F_i=F_i^\top$. A sufficient condition for
\begin{equation}
    f_1(x)\le 0 \implies f_2(x)\le 0
\end{equation}
is the existence of a scalar $\tau\ge0$ such that
\begin{equation}
    \tau
    \begin{bmatrix}
        F_1 & q_1\\
        q_1^\top & c_1
    \end{bmatrix}
    -
    \begin{bmatrix}
        F_2 & q_2\\
        q_2^\top & c_2
    \end{bmatrix}
    \succeq 0.
    \label{eq:s_procedure}
\end{equation}
If there exists $\bar x$ such that $f_1(\bar x)<0$, then, for a single quadratic constraint, this condition is also necessary.
\end{lemma}

In this paper, we use the $\mathcal S$-procedure to certify containment between zero-centered ellipsoids. The following corollary is obtained by applying \Cref{lem:s_procedure} to
$f_i(x)=x^\top Q_i x-q_i$.

\begin{corollary}[Centered ellipsoidal containment]\label{cor:centered_ellipsoid_containment}
Let $Q_1,Q_2\in\mathbb R^{n\times n}$ be positive definite matrices and let 
$q_1,q_2>0$. Define
\begin{equation}
    \begin{aligned}
    \mathcal E_1=\{x\in\mathbb R^n \mid x^\top Q_1x\le q_1\},
    \\
    \mathcal E_2=\{x\in\mathbb R^n \mid x^\top Q_2x\le q_2\}.
    \end{aligned}
\end{equation}
Then $\mathcal E_1\subseteq\mathcal E_2$ if and only if there exists a scalar 
$\tau\ge0$ such that
\begin{subequations}
\begin{align}
    \tau Q_1-Q_2 &\succeq 0, \label{eq:centered_sproc_a}\\
    q_2-\tau q_1 &\ge 0. \label{eq:centered_sproc_b}
\end{align}
\end{subequations}
\end{corollary}

%\subsection{Farkas' Lemma}
\begin{lemma}[Farkas' Lemma]
Let $\cP_1 = \{x \mid G_1x \leq h_1\}$ and $\cP_2 = \{x \mid G_2x \leq h_2\}$ be two polytopes, If $\cP_1 \subseteq \cP_2$ then there is a nonnegative matrix $Z$ such that \cite{miller2023data}
\begin{equation}
    \label{eq:farkas_lemma}
    ZG_1 = G_2, \quad Zh_1 \leq h_2.
\end{equation}
\end{lemma}

\section{Koopman Control Parametrization}
\label{sec:kocop}
This section reviews the main theoretical derivations from initial versions of this work \cite{ondogan2025koopman}, with the introduction of explicit $K_Bu$ term governing the linear contribution of the control directly.
We consider the Koopman system of the form 
\begin{equation}
    \vpsi_x^+ = K_{xx} \vpsi_x + K_Bu + K_{xu}(S\vpsi_x \kron u),
    \label{eq:bilinear_koopman_model}
\end{equation}
describing \eqref{eq:ctrl_affine_system}, 
where $K_{xx}$ is the finite dimensional Koopman matrix governing the drift of the system, $K_B$ matrix introducing the linear control and $K_{xu}$ term governing the cross coupling terms introducing bilinearity to the system, and finally $S$ being a selection matrix that selects a subset of the features, addressing the dimensionality increase caused by the Kronecker product operation. We also introduce a feedback observable map $\vpsi_u:\mathbb R^{n_x}\to\mathbb R^{n_{\psi_u}}$ and seek a control law that is linear in these observables:

\begin{equation}
    \label{eq:linear_ctrl}
    u = K_u \vpsi_u    
\end{equation}

\begin{assumption}
\label{assumption:observable}
There exist a matrix $H$ such that we can model the bilinearity between selected features of $\vpsi_x$ and $\vpsi_u$ into a single linear operation     
    \begin{equation}
        S\vpsi_x \kron \vpsi_u = H\vpsi_x,   
        \label{eq:assumption1}
    \end{equation}
Furthermore, we assume there exist a matrix $C$ that constructs the linear relation
    \begin{equation}
        \vpsi_u = C\vpsi_x.   
        \label{eq:assumption2}
    \end{equation}
\end{assumption}

The assumption in \eqref{eq:assumption1} is not necessarily a linear approximation, it can be constructed as exact as demonstrated in the following example. 
\begin{example}
    $S = \bmat{1 & 0 & 0}$, $\vpsi_x = \bmat{x & x^2 & x^3}^\top$, $\vpsi_u = \bmat{x & x^2}^\top$, satisfies \eqref{eq:assumption1} with $H = \bmat{0 & 1 & 0 \\ 0 & 0 & 1}$ and satisfies \eqref{eq:assumption2} with $C = \bmat{1 & 0 & 0 \\ 0 & 1 & 0}$.
\end{example}

\begin{theorem}
    Under the assumptions \ref{eq:assumption1} and \ref{eq:assumption2}, there exists a closed loop linearized form
    \begin{equation}
        \label{eq:closed_loop_linear}
        \begin{aligned}
        \vpsi_x^+ &= \tilde K \vpsi_x,
        \end{aligned}
    \end{equation}
    where,
    \begin{equation}
        \label{eq:closed_loop_linear_matrix}
        \begin{aligned}
        \tilde K &= K_{xx} + K_BK_uC + K_{xu}(I \kron K_u)H.
        \end{aligned}
    \end{equation}
\end{theorem}

\begin{proof}
    The proof closely follows the derivation in \cite{ondogan2025koopman}. Substituting the control law \eqref{eq:linear_ctrl} and the assumption \eqref{assumption:observable} 
    into \eqref{eq:bilinear_koopman_model}
    \begin{equation}
    \begin{aligned}
        \vpsi_x^+ =& K_{xx}\vpsi_x + K_BK_u\vpsi_u + K_{xu}(S\vpsi_x \kron K_u \vpsi_u) \\
        &K_{xx}\vpsi_x + K_BK_u\vpsi_u + K_{xu}(I \kron K_u)(S\vpsi_x \kron \vpsi_u)\\
        &K_{xx}\vpsi_x + K_BK_uC\vpsi_x + K_{xu}(I \kron K_u)H\vpsi_x\\
        &\big(K_{xx} + K_BK_uC + K_{xu}(I \kron K_u)H\big)\vpsi_x
    \end{aligned}
    \end{equation}
    completes the proof.
\end{proof}

\begin{remark}
    \label{remark:Bu_term}
    While the prior work \cite{ondogan2025koopman} does not include an explicit linear input term $Bu$, instead, it represents linearity through the bilinear component by a constant observable $\vpsi_i = 1$. Let the block matrix $K_{xu}^{(1)}$ denote the block matrix that selects first $n_u$ columns of $K_{xu}$. This implies, $K_{xu}^{(1)}(1 \kron u) = K_{xu}^{(i)}u$ which implies $K_B = K_{xu}$. In contrast, this model considers $K_Bu$ term explicitly, and therefore does not require a constant observable, while including $K_Bu$ term keeps the equations clearer,  
    for representing direct linear control effects.
\end{remark}

\subsection{System Identification of the Koopman Control Parametrization}

In this section, we review an EDMD-like procedure to identify the system matrices 
$K_{xx}, K_B, K_{xu}, H, C,$ and $S$ from data. Using arbitrary controllers, we collect 
a dataset consisting of state-input pairs $(x_t, u_t)$, and define
\begin{equation}
    \mathcal{D} = \{(x_t, u_t)\}_{t=1}^N.    
\end{equation}
We denote the state dataset by
\begin{equation}
    \mathcal{D}_x = \{x_t\}_{t=1}^N,    
\end{equation}
and define the lifted datasets
\begin{equation}
    \mathcal{D}_{\vpsi_x} = \{\vpsi_x(x_t)\}_{t=1}^N, \quad
    \mathcal{D}_{\vpsi_u} = \{\vpsi_u(x_t)\}_{t=1}^N.    
\end{equation}
We further define the paired lifted dataset
\begin{equation}
    \mathcal{D}_{\vpsi} = \{(\vpsi_x(x_t), \vpsi_u(x_t))\}_{t=1}^N.    
\end{equation}
The matrix $C$ can be identified independently of the other matrices by solving the least-squares problem
\begin{equation}
    \min_C \sum_{(\vpsi_x, \vpsi_u) \in \mathcal{D}_{\vpsi}}
    \left\| \vpsi_u - C\,\vpsi_x \right\|^2.
\end{equation}
The variables $H$ and $S$ on the other hand, are coupled, therefore \eqref{eq:assumption2} cannot be solved directly via least squares fit. To address this, we define a binary vector $s \in \mathbb{R}^{n_{\vpsi_x}}$ that represents the selection matrix with an elementwise multiplication
\begin{equation}
    S\vpsi_x = s \odot \vpsi_x.
\end{equation}
Therefore, one can map $S$ to $s$ as
\begin{equation}
    s_j = \sum_i S_{i,j}
\end{equation}
With that we start by assuming $s = \mathbf{1}_{n_{\vpsi_x}}$. That is, initially all the observables are selected, yielding the problem of finding $H$:
\begin{equation}
    \min_H \sum_{(\vpsi_x, \vpsi_u) \in \cD}\|\vpsi_{x} \kron \vpsi_{u} - H\vpsi_{x}\|
\end{equation}
Now consider the $i$-th entry of the selection vector $s_i$, selecting the scalar $\vpsi_{x,i}$, and let $H^{(i)}$ be defined as a block matrix consisting of rows from $in_{\vpsi_x}$ to $(i+1)n_{\vpsi_x}$
with the corresponding expression:
\begin{equation}
    [\vpsi_x]_i \vpsi_u
    -
    H^{(i)}\vpsi_x,
    \label{eq:H_block}
\end{equation}

Note that \eqref{eq:H_block} decouples \eqref{eq:assumption1} into independent individual blocks where $i$-th entry is selected separately then other entries.

Now we turn off the term $s_i$ based on a threshold $\epsilon_{th}$:
\begin{equation}
s_i =
\begin{cases}
0, & \text{if }     [\vpsi_x]_i \vpsi_u
    -
    H^{(i)}\vpsi_x > \epsilon_{\text{th}}, \\
1, & \text{otherwise}.
\end{cases}
\end{equation}

Finally, we refit $H$ based on the new $s$:
\begin{equation}
    \min_H\sum_{(\vpsi_x, \vpsi_u) \in \cD} \|s \odot \vpsi_{x} \kron \vpsi_{u} - H\vpsi_{x}\|
    \label{eq:H_id}
\end{equation}

Once the matrices $C, H, S$ are identified, the rest become the identification of Koopman matrices $K_{xx}, B, K_{xu}$ from the data as
\begin{equation}
\begin{aligned}
\min_{K_{xx}, B, K_{xu}} \quad 
& \sum_{(\vpsi_x, \vpsi_u) \in \cD} 
\Big\| \vpsi_{x,t+1} \\
&\quad - \begin{bmatrix} K_{xx} & K_B & K_{xu} \end{bmatrix}
\begin{bmatrix}
\vpsi_{x,t} \\
u_t \\
S \vpsi_{x,t} \kron u_t
\end{bmatrix}
\Big\|.
\end{aligned}
\label{eq:koopman_id}
\end{equation}

\subsection{Lyapunov Function Selection and Control Design via Alternating Optimization for the Koopman Control Parametrization}
We consider a quadratic Lyapunov candidate function in the observables $V(\vpsi_x) = \vpsi_x^\top P \vpsi_x$ with $P \succ 0$. With the help of Schur complement, optimization problem can be formulated to stabilize the linear closed loop system equation \eqref{eq:closed_loop_linear}

\begin{equation}\label{eq:sdp}
  \begin{aligned}
    \min_{\lambda,K_u,P}\,& \lambda\\
    \subjectto&
                \begin{aligned}
                  &\bmat{P & P \tilde{K} \\
                  \tilde{K}^\top P & \lambda P}\succeq 0.\\
                  &\lambda\in[0,1)
                \end{aligned}
  \end{aligned}
\end{equation}

\begin{remark}
    Since the stability condition in \eqref{eq:sdp} includes the closed-loop matrix $\tilde K$, the assumptions \eqref{eq:assumption1}, \eqref{eq:assumption2} may contain modeling errors. In \Cref{sec:appendix} we include the validity of the Lyapunov condition under these modeling mismatches.
\end{remark}

However the problem \eqref{eq:sdp} is biconvex in this form as it contains the matrix variable $P$ in multiplications of variables $\lambda$ and $\tilde K$, where $\tilde K$ contains the optimization variable $K_u$ in it, as \eqref{eq:closed_loop_linear_matrix} illustrates. The proposed algorithm by \cite{ondogan2025koopman} tries random $P$ matrices until feasibility by the following method. Given a random matrix $R$, it generates a positive semidefinite matrix $R^\top R$. Later, it considers  the linear decoding operator $A_{\textrm{dec}}$ as,
\begin{equation}
    x = A_{\textrm{dec}}\vpsi_x
\end{equation}
It then proposes a projected matrix $A_{dec}^\top R^\top RA_{dec}$ that define a quadratic candidate Lyapunov function in the state space. Finally, a small $\epsilon$ is added on the diagonals to ensure positive definiteness, yielding the candidate
\begin{equation}
    P =A_{dec}^\top R^\top RA_{dec} + \epsilon I.
\end{equation}

with this candidate $P$, we obtain convex program
\begin{equation}\label{eq:sdp_Ku}
  \begin{aligned}
    \min_{\lambda,K_u}\,& \lambda\\
    \subjectto&
                \begin{aligned}
                  &\bmat{P & P \tilde{K} \\
                  \tilde{K}^\top P & \lambda P}\succeq 0.\\
                  &\lambda\in[0,1)
                \end{aligned}
  \end{aligned}
\end{equation}

\subsubsection{Alternate Convex Search}
\label{sec:alternatingConvexOptimization}
Once a feasible $K_u$, and $P$ pair is obtained by the procedure described above, we introduce an alternating method for addressing the biconvexity of the problem. In particular, we formulate the problem as an Alternate Convex Search (ACS) problem \cite{zhu2025disciplined}. We now solve for a new $P$ under the fixed $K_u$ and $\lambda'$, where $\lambda'$ is the new candidate decay constant with $\lambda' < \lambda$.

\begin{equation}\label{eq:sdp_P}
  \begin{aligned}
    \min_{P}\,& \lambda'\\
    \subjectto&
                \begin{aligned}
                  &\bmat{P & P \tilde{K} \\
                  \tilde{K}^\top P & \lambda' P}\succeq 0.\\
                \end{aligned}
  \end{aligned}
\end{equation}

The variable $\lambda'$ however, cannot be made an optimization variable since again it violates convexity due to term $\lambda'P$. Therefore we adopt a bisection method that exhaustively tries different $\lambda'$ candidates until a feasible solution is found. Once a feasible $P$ is found, we can alternately solve for \eqref{eq:sdp}. The procedure can be followed by alternating solutions of \eqref{eq:sdp} and \eqref{eq:sdp_P}. 

\begin{algorithm}[t]
\caption{Alternating Convex Optimization with Bisection on $\lambda$}
\label{alg:aco_lambda}
\begin{algorithmic}[1]
\State \textbf{Input:} Initial feasible $(K_u, P, \lambda)$, $N_{\textrm{outer}}, N_{\textrm{inner}}$ and threshold $\varepsilon$

\For{$k = 0,1,\dots,N_{\text{outer}}-1$}
    \State \textbf{Step 1 (Bisection on $\lambda$):}
    \State Initialize bounds:
    \[
    \lambda_{\min} \gets 0, \quad \lambda_{\max} \gets \lambda
    \]
    
    \For{$i = 0,1,\dots,N_{\text{inner}}-1$}
    
        \State $\lambda' \gets \frac{1}{2}(\lambda_{\min} + \lambda_{\max})$

        \State Solve SDP \eqref{eq:sdp_P} with $(K_u, \lambda')$
        
        \If{feasible}
            \State $\lambda_{\max} \gets \lambda'$
            \State \textbf{Lyapunov matrix update}:
            \State Update $P$ with the solution
        \Else
            \State $\lambda_{\min} \gets \lambda'$
        \EndIf
        
        \If{$|\lambda_{\max} - \lambda_{\min}| \leq \varepsilon$}
            \State \textbf{break}
        \EndIf
        
    \EndFor

    \State \textbf{Step 2 (Controller update):}
    \State Solve SDP \eqref{eq:sdp_Ku} to obtain $(K_u, \lambda)$
\EndFor
\State \textbf{Output:} $(K_u, P, \lambda)$
\end{algorithmic}
\end{algorithm}

\section{Control Under Input Constraints}
\label{sec:constraints}
While \eqref{eq:sdp} offers a control that stabilizes the system, 
this section is dedicated to provide satisfaction of control input limits, as long as lifted state remains in a predefined set. In this section we introduce two ways of incorporating actuator saturations into control by conditioning them to convex regions of states. Firstly, we consider ellipsoidal constraints via an application of the $\cS$-procedure  \cite{polik2007survey}; which introduces an additional semi-definite constraint. Secondly, we consider polytopic constraints via an application of the Farkas' lemma;  this formulation has the advantage of resulting in linear constraints that are computationally faster. 

We first develop the conditional relation between two sets. Given that a set in embedded space $\cP_{\vpsi}$, and a set in the input space, $\cP_u$, the task can be formulated as the satisfaction of the condition
\begin{equation}
    \vpsi_u\in\mathcal P_\psi \implies K_u\vpsi_u\in\mathcal P_u.
    \label{eq:set_containment}
\end{equation}

In many control applications, however, the admissible region is specified in the original state space rather than in the lifted observable space. That is, one is given a set $\mathcal{P}_x$ of allowable states and wishes to guarantee $u\in\mathcal{P}_u$ for all $x\in\mathcal{P}_x$.
Following this condition, we extend the method such that, given a set in state space $\cP_{x}$ (instead of in lifted space $\cP_{\vpsi}$), we ensure satisfaction of constraints in input space $\cP_u$.

\begin{example}
    Suppose a state space is given in $\mathbb{R}^2$, and we have a single lifting function $\vpsi(x_1, x_2):\mathbb{R} \times \mathbb{R} \to \mathbb{R}
    $, and together with the state, defining the Koopman observable vector $\vpsi_x(\vx) = \bmat{x_1 & x_2 & \vpsi(x_1, x_2)}^\top$. Additionally, suppose a $\cP_x $ be given as a polytope in the state space $\{x \mid Ax-b \leq 0\}$. Fig.~\ref{fig:lifting_polytope} demonstrates the representative polytope on a 2D surface in states $(x_1, x_2)$. The problem then becomes finding the polytope in the 3D space, including the datapoints of triplets $\big(x_1, x_2, \vpsi(x_1, x_2)\big)$.
\end{example}

\begin{figure}[h!]
    \centering
    \includegraphics[width=1\linewidth]{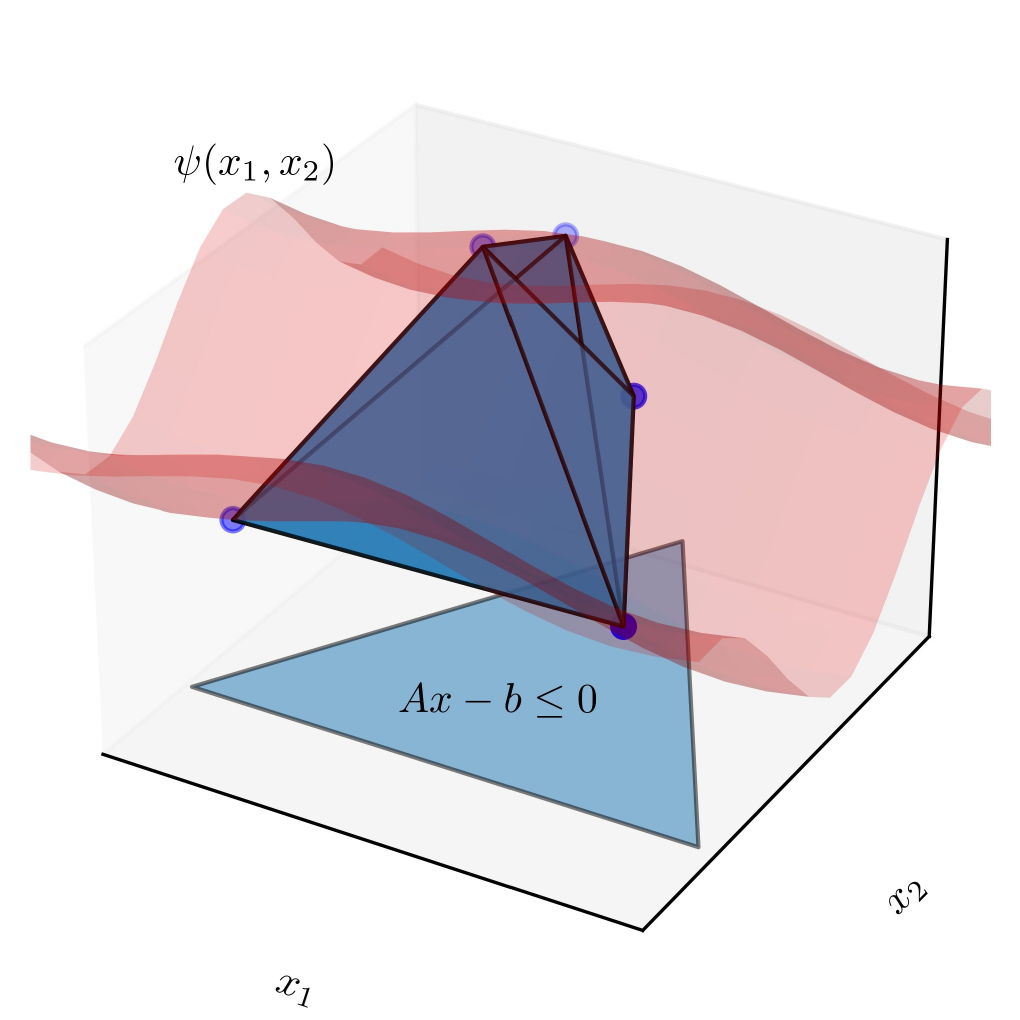}
    \caption{The procedure of finding polytope in the lifted space is exemplified in a 2D state-space with a single lifting function $\vpsi(x_1, x_2)$.}
    \label{fig:lifting_polytope}
\end{figure}

\subsection{Ellipsoidal Constraints}
Let two ellipsoids, parametrized by positive semidefinite matrices $Q_{\vpsi} \in \mathbb{R}^{n_{\vpsi_u} \times n_{\vpsi_u}}$ and $Q_u \in \mathbb{R}^{n_u \times n_u}$ and two positive scalars $q_u, q_{\vpsi} \in \mathbb{R}$ given as
$\cP_{\vpsi}(Q_{\vpsi},q_{\vpsi}) \defeq \{\vpsi_{u} \mid \vpsi_u^\top Q_{\vpsi} \vpsi_u \leq q_{\vpsi}\},$ and $\cP_u(Q_u, q_u) = \{u \mid u^\top Q_u u \leq q_u\}$, we seek to find a controller that is input-constrained, whenever the observables are constrained. In other words, under the controller $ K_u$, the relation \eqref{eq:set_containment} should be satisfied.

\subsubsection{Ellipsoidal Constraints in the Embedding Space}

The ellipsoidal constraint in the input space can equivalently be expressed as the following implication:
\begin{equation}
\vpsi_u^\top Q_{\vpsi} \vpsi_u \leq q_{\vpsi}
\quad\Longrightarrow\quad
\vpsi_u^\top K_u^\top Q_{u} K_u \vpsi_u \leq q_{u}.
\label{eq:ellipsoid_constraints}
\end{equation}

Applying the $\cS$-procedure, condition \eqref{eq:ellipsoid_constraints} implies that there exists a $\tau \geq 0$ such that
\begin{subequations}
\begin{equation}
    \tau Q_{\vpsi} - K_u^\top Q_{u} K_u \succeq 0
    \label{eq:S-proc_a} 
\end{equation}
\begin{equation}
    q_{u} - \tau q_{\vpsi} \geq 0.
    \label{eq:S-proc_b}
\end{equation}
\end{subequations}
Refactoring the first constraint using Schur complements, and incorporating it in the design problem we obtain the following optimization problem:
\begin{equation}
\label{eq:controller_ellipsoid_constraints}
\begin{aligned}
    \min_{\lambda, K_u, \tau, P} &\lambda \\
    \subjectto  &\bmat{P & P\tilde K \\
                  \tilde K^\top P & \lambda P}\succeq 0 \\
                &\bmat{\tau Q_{\vpsi} & K_u^\top Q_{u} \\
                  Q_{u} K_u & Q_{u}}  \succeq 0, \\
               &q_{u} - \tau q_{\vpsi} \geq 0.
\end{aligned}
\end{equation}
where the biconvexity is handled by ACS formulation as in \Cref{sec:alternatingConvexOptimization}

\subsubsection{Ellipsoidal Constraints in State Space}
While condition \eqref{eq:ellipsoid_constraints} provides the sets in the observable space, here we develop the same formulation where the set is given in the state space. Suppose that the constraint is defined as $x^\top Q_x x \leq q_x$, since the map $\vpsi_u(x)$ is nonlinear, it does not necessarily impose an ellipsoidal shape on the observables.
In that case, we propose the following procedure to design ellipsoidal constraints as an approximation:
given $x^\top Q_x x \leq q_x$ and the dataset that consists of states $\cD_x$, we consider the subset of points that are inside the ellipsoid as \eqref{eq:subset_dataset_ellipsoid}. Following that, we lift the points in this intersection, yielding a lifted dataset \eqref{eq:subset_dataset_ellipsoid_lifted}
\begin{subequations}
\begin{equation}
    \cD_x \cap \{x \mid x^\top Q_x x \leq q_x\}\label{eq:subset_dataset_ellipsoid},
\end{equation}
\begin{equation}
    \cD_{\vpsi_u} \cap \{\vpsi_u(x) \mid x^\top Q_xx \leq q_x\}. \label{eq:subset_dataset_ellipsoid_lifted}
\end{equation}
\end{subequations}
On the lifted dataset, we find a minimum-volume ellipsoid \cite{boyd2004convex} as:
\begin{equation}
\begin{aligned}
    \label{eq:MVE}
    Q_{\vpsi}^\star \in \argmax_{Q_{\vpsi} \succ 0} \quad & \log\det(Q_{\vpsi}) \\
    \subjectto \quad 
    & \vpsi_u^\top Q_{\vpsi} \vpsi_u \leq 1,
    \quad \\
    &\forall \vpsi_u \in \cD_{\vpsi_u} \cap \{\vpsi_u(x) \mid x^\top Q_xx \leq q_x\}.
\end{aligned}
\end{equation}
\begin{remark}
    $q_{\vpsi}$ is chosen as 1 to address the scale ambiguity.
\end{remark}

Consequently, we obtain the ellipsoidal constraint in the observable space as $\cP_{\vpsi}(Q_{\vpsi}^\star, 1)$ that is, the mimumum-volume ellipsoid that containts all the lifted datapoints.

\begin{remark}
    This construction provides an outer approximation of the lifted image of the admissible state set over the available dataset.
\end{remark}

\subsection{Polytopic Constraints}
\subsubsection{Polytopic constraints in embedding space}
We now consider the control input to be constrained by a polytope $\mathcal{U} = \{u \in \mathbb{R}^{n_u}| Gu \leq g\}$ whenever the observables is in a polytope $\mathcal{P} = \{\vpsi_u |A\vpsi_u  \leq b\}$,
then following \eqref{eq:farkas_lemma} there exists a non-negative matrix $Z\geq 0$ such that
\begin{equation}
\begin{aligned}
    G K_u &= Z A,\\
    Zb &\leq g;
    \end{aligned}
    \label{eq:control_farkas}
\end{equation}
this constraint is linear and convex. Adding \eqref{eq:control_farkas} to the design problem \eqref{eq:sdp} leads to the following optimization problem 
\begin{equation}
\label{eq:controller_polytopic_constraints}
\begin{aligned}
    \min_{\lambda, K_u, Z} &\lambda \\
    \subjectto & \bmat{\lambda P & \tilde K P \\ P \tilde K & P} \succeq 0 \\
               & G K_u = Z A \\
               & Zb \leq g \\
               & Z \geq 0
\end{aligned}
\end{equation}
Similarly, we can use alternating optimization to handle biconvexity.

\subsubsection{Polytopic Constraints in the State Space via Box Approximation}
Similarly, the polytopic constraints are often provided in state space. 

In the special case when the states are constrained within a box
$
\mathcal{P} = \{x \in \mathbb{R}^{n_x} \mid x_{i,\min} \le x_i \le x_{i,\max}, \; i = 1,\dots,n_x\},
$
the bounds on each feature $\vpsi_i(x)$ can be obtained by solving
\begin{equation}
\min_{x \in \mathcal{P}} \vpsi_i(x) \;\le\; \vpsi_i(x) \;\le\; \max_{x \in \mathcal{P}} \vpsi_i(x).
\end{equation}

If we have $\mathcal{P} = \{x:A_xx - b_x \leq 0\}$ using the data set $\cD_x$, we can conservatively define an approximate axis aligned box constraint,
\begin{equation}
\min_{x \in \cD_x \cap \cP} \vpsi_i(x) \;\le\; \vpsi_i(x) \;\le\; \max_{x \in \cD_x \cap \cP} \vpsi_i(x).
\label{eq:box_contraint}
\end{equation}

\subsubsection{Polytopic Constraints in the State Space via Robust Optimization}
While \eqref{eq:box_contraint} is conservative, as it fits a larger polytope than required, we define a robust linear programming based formulation to obtain a tighter polytope in the observable space from the data. 

Starting from the original datasets in state space $\cD_x$ and embedding space $\cD_{\vpsi_u}$, we first define the set of points of these datasets that lie inside of the polytope as
\begin{subequations}
   \begin{equation}
       \bar{\mathcal{D}}_x \defeq \mathcal{D}_x \cap \{x \mid A_xx -b_x \leq 0\},
   \end{equation} 
   \begin{equation}
       \bar{\mathcal{D}}_{\vpsi} \defeq \mathcal{D}_{\vpsi_u} \cap \{\vpsi_u(x) \mid A_xx - b_x \leq 0\}
   \end{equation}
    \label{eq:dataset_intersection}
\end{subequations}

 We consider the data matrices $X = \bmat{x_1, x_2 \dots x_N}$ generated by the dataset $\bar{\cD}_x$ and $\vpsi = \bmat{\vpsi_1, \vpsi_2, \dots \vpsi_N}$ generated by $\bar{\cD}_{\vpsi_u}$ correspondingly with the columns are aligned such that $\vpsi_u(x_i) = \vpsi_i$. 

To relate the state space and the observable space using only data, we adopt the following assumption, which enforces an approximate consistency of the lifting map over convex combinations of the data points.
\begin{assumption}
    $p \in \mathbb{R}^{N}$ with $p\geq0$ and  $\mathbf{1}^\top p = 1$, $\vpsi p \approx \vpsi(Xp)$ 
    \label{assumption:convex_comb}
\end{assumption}

\begin{theorem}[Robust Linear Program]
    Under the assumption \ref{assumption:convex_comb}, the constraints
    \begin{equation}
    \begin{aligned}
    & z_j^\top b + \lambda_j \leq g_j \\
    & X^\top A^\top z_j + \lambda_j \mathbf{1} \ge \vpsi^\top K_u^\top G_j^\top \\
    & z_j \ge 0.
    \label{eq:Robust_LP_constr}
    \end{aligned}
    \end{equation}

    implies the condition
    \begin{equation}
        A_xx - b_x \implies Gu - g
    \end{equation}
    where the input is controlled by \eqref{eq:linear_ctrl}.
\end{theorem}

\begin{proof}
The proof starts by considering the worst case $p$, that maximizes the multiplication in the each row $G_ju$ and makes sure that it remains within the polytope by comparing it to $g_j$, under the constraints. replacing $u$ with $K_u\vpsi_u$  
\begin{equation}
\begin{aligned}
\bmat{\max_p &G_j K_u\vpsi_u\\
\subjectto& p\transpose 1=1\\
& p\geq 0\\
& \vpsi p = \vpsi_u\\
& X p = x\\
&A_xx \leq b} \leq g_j,
\label{eq:conv_env}
\end{aligned}
\end{equation}

replacing the equalities constraints, we obtain

\begin{equation}
\begin{aligned}
\bmat{\max_p &G_j K_u\vpsi p\\
\subjectto& 
A_xXp \leq b_x \\
&p\transpose 1=1\\
& p\geq 0\\
} \leq g_j
\end{aligned}
\end{equation}

We introduce dual variables $z_j \ge 0$ associated with the inequality constraint $A_xXp \le b$, and $\lambda_j \in \mathbb{R}$ associated with the equality constraint $\mathbf{1}^\top p = 1$.

The Lagrangian is given by
\begin{equation}
\begin{aligned}
\mathcal{L}(p, z_j, \lambda_j)
= G_j K_u \vpsi p + z_j^\top (b_x - A_xXp) + \lambda_j (1 - \mathbf{1}^\top p).
\end{aligned}
\end{equation}

Rearranging terms, we obtain
\begin{equation}
\begin{aligned}
\mathcal{L}(p, z_j, \lambda_j)
&= \big( \vpsi^\top K^\top G_j^\top - X^\top A_x^\top z_j - \lambda_j \mathbf{1} \big)^\top p \\
&+ z_j^\top b_x + \lambda_j.
\end{aligned}
\end{equation}

For the supremum over $p \ge 0$ to be finite, the coefficient of $p$ must be non-positive, which yields the dual feasibility condition
\begin{equation}
X^\top A_x^\top z_j + \lambda_j \mathbf{1} \ge \vpsi^\top K_u^\top G_j^\top.
\end{equation}

Therefore, the dual problem is given by
\begin{equation}
\left[
\begin{aligned}
\min_{z_j,\lambda_j}\quad 
    & z_j^\top b_x + \lambda_j \\
\text{s.t.}\quad 
    & X^\top A_x^\top z_j + \lambda_j \mathbf{1}
      \ge \vpsi^\top K_u^\top G_j^\top, \\
    & z_j \ge 0
\end{aligned}
\right]
\leq g_j .
\end{equation}

If the minimum is smaller, it can be dropped, yielding the three constraints in \eqref{eq:Robust_LP_constr}, thus completing the proof.

\end{proof}

Finally, the constrained controller is obtained by solving 
\begin{equation}
\label{eq:controller_robust_lp}
\begin{aligned}
    \min_{\lambda, K_u, z_j} &\lambda \\
    \subjectto & \bmat{\lambda P & \tilde K P \\ P \tilde K & P} \succeq 0 \\
    & z_j^\top b + \lambda_j \leq g_j \\
    & X^\top A^\top z_j + \lambda_j \mathbf{1} \ge \vpsi^\top K_u^\top G_j^\top \\
    & z_j \ge 0.
\end{aligned}
\end{equation}

\section{Stabilization around a non-zero equilibrium}\label{sec:setpoint}
In this section, we augment the \kocop{} controller with a feedforward term to regulate the system around a given setpoint $x_t^{\textrm{ref}}$. Specifically, although \eqref{eq:linear_ctrl} offers a stabilizing controller around the origin, we extend our control law to  stabilize the error dynamics around an arbitrary setpoint.

\subsection{Error Dynamics Around a Nonzero Equilibrium}
We define the lifted terms corresponding to $\vpsi_x(x^{\textrm{ref}}_t) = \vpsi_{x,t}^{\textrm{ref}}$ and $\vpsi_u(x^{\textrm{ref}}_t) = \vpsi_{u,t}^{\textrm{ref}}$, and introduce the feedforward term $k_{\textrm{ff}}$, yielding the
feedback-feedforward control law
\begin{equation}
    u_t = K_u \big(\vpsi_{u,t} - \vpsi_{u,t}^{\textrm{ref}}\big) + k_{\textrm{ff}},
    \label{eq:feedforward_control_law}
\end{equation}

where $K_u$ is the precomputed controller that stabilizes the origin, but applied to the lifted error term $\vpsi_{u,t} - \vpsi_{u,t}^{\textrm{ref}}$, and the term $k_{\textrm{ff}}$ is chosen to provide the required steady-state input.
Such feedback-feedforward decompositions are standard for linear systems where they yield a linear error dynamics \cite{friedland2012control} \cite{limon2008mpc}. Here we derive an analogous construction to the bilinear Koopman systems.
We rewrite the closed loop bilinear dynamics under the controller \eqref{eq:feedforward_control_law} as
\begin{equation}
    \begin{aligned}
    \vpsi_{x, t+1} &= K_{xx}\vpsi_{x,t} + K_BK_u(\vpsi_{u,t}- \vpsi_u^{\textrm{ref}}) 
    \\&+ K_{xu}(S\vpsi_{x,t} \kron K_u(\vpsi_{u,t} - \vpsi_u^{\textrm{ref}})) \\&+K_B k_{\textrm{ff}} + K_{xu}(S \vpsi_x \kron k_{\textrm{ff}})
    \end{aligned} \label{eq:controlled_bilinear_system}
\end{equation}

We then define the error states $e_{x,t} = \vpsi_{x,t} - \vpsi^\textrm{ref}_{x,t}$ and $e_{u,t} = \vpsi_{u,t} - \vpsi^\textrm{ref}_{u,t}$. This subtraction, along with \eqref{eq:feedforward_control_law} yields the error equation for the next state as
\begin{equation}
\begin{aligned}
    e_{x, t+1} =&  
    K_{xx}\vpsi_{x,t} + K_BK_ue_{u,t} + K_{xu}(S\vpsi_x \kron K_u e_{u,t})\\
    &+ K_{xu}(S\vpsi_x \kron k_{\textrm{ff}}) + K_B k_{\textrm{ff}} - \vpsi_{x,t+1}^{\textrm{ref}}.
    \label{eq:error_next}
\end{aligned}
\end{equation}

\subsection{Stabilization Around a Known Constant Equilibrium}
Given a constant setpoint reference $\vpsi_x^{\textrm{ref}} \defeq \vpsi_{x,t}^{\textrm{ref}} = \vpsi_{x,t+1}^{\textrm{ref}}$, and its corresponding potentially unknown constant input $u^{\textrm{ref}}$, the steady-state Koopman dynamics satisfy
\begin{equation}
    \vpsi_x^{\textrm{ref}} = K_{xx} \vpsi_x^{\textrm{ref}} + K_Bu^{\textrm{ref}} + K_{xu}(S\vpsi_x^{\textrm{ref}} \kron u^{\textrm{ref}}). 
    \label{eq:reference_point}
\end{equation}

Throughout this section, we first assume the equilibrium input $u^{\textrm{ref}}$ is known for the corresponding setpoint $\vpsi^{\textrm{ref}}$, and derive the closed-loop dynamics. Then we extend it to the case when it is not known.

Using the equilibrium dynamics \eqref{eq:reference_point}, we replace the next reference observable $\vpsi_{x,t+1}^{\textrm{ref}}$
term in \eqref{eq:error_next}. This yields the error dynamics in the embedded space, centered around the constant reference point as
\begin{equation}
\begin{aligned}
\label{eq:error_dynamics_prelim}
e_{x,t+1} &=
K_{xx}e_{x,t} + K_BK_ue_{u,t}
\\&+ K_{xu}(S\vpsi_{x,t} \kron K_u e_{u,t}) 
\\&+ K_{xu}(S\vpsi_{x,t} \kron k_{\textrm{ff}})
- K_{xu}(S\vpsi^{\textrm{ref}}_{x} \kron u^{\textrm{ref}})
\\&+ K_Bk_{\textrm{ff}} - K_Bu^{\textrm{ref}}.
\end{aligned}
\end{equation}
To obtain an explicit form of the closed-loop error dynamics, we first add and subtract $K_{xu}(S\vpsi_x^{\textrm{ref}} \kron K_u \vpsi^{\textrm{ref}}_u)$, and rewrite the term $K_{xu}(S\vpsi_{x,t} \kron Ke_{x,t})$ as $K_{xu}(S\vpsi_{x,t} \kron K_u \vpsi_{u,t}) - K_{xu}(S\vpsi_{x,t} \kron K_u \vpsi_u^{\textrm{ref}})$
\begin{equation}
\begin{aligned}
e_{x,t+1} &=
K_{xx}e_{x,t} + K_BK_u(\vpsi_{u,t} - \vpsi_u^{\textrm{ref}})
\\&+ K_{xu}(S\vpsi_{x,t} \kron K_u \vpsi_{u,t}) - K_{xu}(S\vpsi_{x,t} \kron K_u \vpsi_u^{\textrm{ref}}) 
\\&+K_{xu}(S\vpsi_x^{\textrm{ref}} \kron K_u \vpsi_u^{\textrm{ref}}) - K_{xu}(S\vpsi_x^{\textrm{ref}} \kron K_u \vpsi_u^{\textrm{ref}})
\\&+ K_{xu}(S\vpsi_{x,t} \kron k_{\textrm{ff}})
- K_{xu}(S\vpsi^{\textrm{ref}}_{x} \kron u^{\textrm{ref}})
\\&+ K_Bk_{\textrm{ff}} - K_Bu^{\textrm{ref}}.
\end{aligned}
\end{equation}
We then use assumptions \eqref{eq:assumption1} and \eqref{eq:assumption2} to combine the terms $K_{xu}(S\vpsi_{x,t} \kron K_u \vpsi_{u,t})$ and $K_{xu}(S\vpsi_x^{\textrm{ref}} \kron K_u \vpsi_u^{\textrm{ref}})$ into $K_{xu}(I \kron K_u)H (\vpsi_{x,t} -\vpsi_x^{\textrm{ref}}) = K_{xu}(I \kron K_u)He_{x,t}$

\begin{equation}
\begin{aligned}
    e_{x,t+1} & =
    K_{xx}e_{x,t} + K_BK_uCe_{x,t} + K_{xu}(I \kron K_u)H e_{x,t}
    \\& - K_{xu}(S\vpsi_{x,t} \kron K_u \vpsi_x^{\textrm{ref}}) + K_{xu}(S_x\vpsi^{\textrm{ref}} \kron K_u \vpsi_u^{\textrm{ref}})  
    \\& - K_{xu}(S\vpsi^{\textrm{ref}}_{x} \kron u^{\textrm{ref}})
\\& - K_Bu^{\textrm{ref}} + K_Bk_{\textrm{ff}} + K_{xu}(S\vpsi_{x,t} \kron k_{\textrm{ff}}). 
\end{aligned}   
\end{equation}
We rewrite this expression using the closed-loop Koopman matrix $\tilde{K}$:
\begin{equation}
\begin{aligned}
    e_{x,t+1} & =
    \tilde K e_{x,t} \\& - K_{xu}(S\vpsi_{x,t} \kron K_u \vpsi_u^{\textrm{ref}}) + K_{xu}(S\vpsi_x^{\textrm{ref}} \kron K_u \vpsi_u^{\textrm{ref}}) \\& - K_Bu^{\textrm{ref}} - K_{xu}(S\vpsi^{\textrm{ref}}_{x} \kron u^{\textrm{ref}}_t) \\&  + K_Bk_{\textrm{ff}} + K_{xu}(S\vpsi_{x,t} \kron k_{\textrm{ff}}).
\label{eq:error_dynamics}
\end{aligned}   
\end{equation}
The terms $- K_{xu}(S\vpsi_{x,t} \kron K_u \vpsi_u^{\textrm{ref}}) + K_{xu}(S\vpsi_x^{\textrm{ref}} \kron K_u \vpsi_u^{\textrm{ref}})$ can be combined into $-K_{xu}(Se_{x,t} \kron K_u\vpsi_u^{\textrm{ref}})$. Using the identity \eqref{eq:kron_identity_to_matrix_vector}, this can then be converted to $-K_{xu}(I \kron K_u \vpsi_u^{\textrm{ref}})e_{x,t}$, indicating the linear relation in $e_{x,t}$. 

\begin{equation}
\begin{aligned}
    e_{x,t+1} & =
    \tilde K e_{x,t} - K_{xu}(I \kron K_u \vpsi_u^{\textrm{ref}})e_{x,t} \\&  - K_Bu^{\textrm{ref}} - K_{xu}(S\vpsi^{\textrm{ref}}_{x} \kron u^{\textrm{ref}}_t) \\&  + K_Bk_{\textrm{ff}} + K_{xu}(S\vpsi_{x,t} \kron k_{\textrm{ff}}).
\end{aligned}   
\end{equation}

We then aim to select $k_{\textrm{ff}}$ that cancels the last four residual terms with the choice of
\begin{equation}
    k_{\textrm{ff}} = u^{\textrm{ref}},
\end{equation}

and obtain
\begin{equation}
    \begin{aligned}
        e_{x,t+1} &=
        \tilde K e_{x,t} - K_{xu}(I \kron K_u \vpsi_u^{\textrm{ref}})e_{x,t} \\&\quad+ K_{xu}(Se_{x,t} \kron k_{\textrm{ff}}). 
    \end{aligned}
\end{equation}
We reuse the identity \eqref{eq:kron_identity_to_matrix_vector} to obtain a linear form in the error state for a fixed $k_{\textrm{ff}}$ 
\begin{equation}
\begin{aligned}
    e_{x,t+1} & =
    \tilde K e_{x,t} - K_{xu}(I \kron K_u \vpsi_u^{\textrm{ref}})e_{x,t} \\&\quad + K_{xu}(k_{\textrm{ff}} \kron I)Se_{x,t}.
\end{aligned}   
\end{equation}

We define new closed-loop Koopman matrix as
\begin{equation}
\begin{aligned}
    \bar K(\vpsi^{\textrm{ref}}, k_{\textrm{ff}}) &= \tilde K - K_{xu}(I \kron K_u \vpsi_u^{\textrm{ref}}) \\&\quad + K_{xu}(k_{\textrm{ff}} \kron I)S
    \label{eq:reference_dependent_closed_loop_matrix}
\end{aligned}
\end{equation}

Therefore we arrive at an expression where the error is linear in a reference- and feedforward term-dependent closed loop matrix \eqref{eq:reference_dependent_closed_loop_matrix}:
\begin{equation}
\begin{aligned}
    e_{x,t+1} &= \bar K(\vpsi_u^{\textrm{ref}}, k_{\textrm{ff}})e_{x,t}
\end{aligned}
\end{equation}

\begin{remark}
    Notice that due to bilinearity, for a fixed $\vpsi^{\textrm{ref}}$, $u^{\textrm{ref}}$ is in linear form.
    If $u^{\textrm{ref}}$ is unknown, it can be estimated from the inverse dynamics of the Koopman equation \eqref{eq:reference_point}. Using the identities above, the steady-state condition can be written as 
    \begin{equation}
        \hat u^{\textrm{ref}} = \big(K_B + K_{xu}T_{n_{\vpsi_x}, n_{u}}(I \kron S\vpsi_{x}^{\textrm{ref}})\big)^\dagger (I -K_{xx})\vpsi_{x}^{\textrm{ref}}
        \label{eq:u_ref_estimation}
    \end{equation}
\end{remark}

\begin{remark}
    In case of instability for a specific $\vpsi^{\textrm{ref}}, k_{\textrm{ff}}$ pair, the controller $K_u$ can be redesigned via modifying \eqref{eq:sdp} using the new closed loop Koopman matrix $\bar K(\vpsi^{\textrm{ref}})$
\begin{equation}\label{eq:sdp_barK}
  \begin{aligned}
    \min_{\lambda,K_u, P}\,& \lambda\\
    \subjectto&
                \begin{aligned}
                  &\bmat{P & P \bar K(\vpsi^{\textrm{ref}}, k_{\textrm{ff}}) \\
                  \bar K(\vpsi^{\textrm{ref}}, k_{\textrm{ff}})^\top P & \lambda P}\succeq 0.\\
                  &\lambda\in[0,1)
                \end{aligned}
  \end{aligned}
\end{equation}
\end{remark}

\begin{remark}
\label{remark constraints}
If the pair $(\vpsi^{\textrm{ref}}, u^{\textrm{ref}})$ is known prior to system identification, the equilibrium information can be incorporated directly into the identification procedure. Specifically, the Koopman model can be identified by solving \eqref{eq:koopman_id} subject to the equilibrium constraint \eqref{eq:reference_point}. This ensures that the identified bilinear Koopman model is exactly consistent with the known equilibrium, thereby eliminating the equilibrium mismatch between the learned model and the true system at this point. Consequently, controllers designed using the identified model can achieve reduced steady-state tracking errors.
\end{remark}

With this constraint, the known equilibrium is enforced exactly in the identified bilinear Koopman model, eliminating modeling error at the equilibrium and thereby reducing the resulting steady-state tracking error.

\subsection{Tracking of a Time-varying Trajectory with a CLF-based Method}\label{sec:clf_tracking}
While a constant feedforward term can stabilize the Koopman system around a constant setpoint, in this section, we extend the use of feedforward term to track a moving reference trajectory. While the feedforward term introduced in previous section was constant, with this CLF formulation, we achieve a time and state dependent selection of feedforward term.
Let $V(e_{x}) = e_x^\top P e_x = \|e_x\|^2_P$ be the CLF function, following \eqref{eq:error_dynamics}, our goal is to find the $k_{\textrm{ff}}$ such that

%\begin{equation}
%    \min_{k_{\textrm{ff}}} V(e_{x, t+1}) - \lambda V(e_{x,t})
%    \label{eq:clf}
%\end{equation}
We use a discrete time formulation of CLF \cite{agrawal2017discrete},
\begin{equation}
    \begin{aligned}
        \min_{k_\textrm{ff}} ~&\|k_\textrm{ff}\|^2_2 + \gamma \delta \\
        \subjectto &V(e_{x, t+1}) -\lambda V(e_{x,t}) \leq \delta \\
        &\delta \geq 0
    \end{aligned}
\end{equation} 
where $\delta$ is the relaxation variable if stability condition is infeasible in a certain timestep, as introduced in \cite{ames2019control}.

We reconsider the error dynamics in \eqref{eq:error_next}, while keeping the next-state reference observables explicitly as $\vpsi_{x,t+1}^{\textrm{ref}}$. This formulation does not require the reference to be a static equilibrium.

Due to bilinearity, for a fixed $\vpsi_{x,t}$ and $e_{x,t+1}$, this form becomes linear in $k_{\textrm{ff}}$, facilitating the control design. To show this, we use \eqref{eq:kron_identity_to_matrix_vector}
and \eqref{eq:commutation_matrix},
\begin{equation}
\begin{aligned}
    e_{x, t+1} =&  
    K_{xx}\vpsi_{x,t} + K_BK_ue_t + K_{xu}(S\vpsi_x \kron K_u e_t) - \vpsi_{x,t+1}^{\textrm{ref}}\\
    &+ \big(B + K_{xu}T_{n_u, n_{s}}(I\kron S\vpsi_x )\big) k_{\textrm{ff}}
    \label{eq:error_next_clf}\\
    \defeq& A_{\textrm{ff}}(\vpsi_{x,t})k_{\textrm{ff}} + b_{\textrm{ff}}(\vpsi_{x,t}),
\end{aligned}
\end{equation}
where 
\begin{subequations}
    \begin{equation}
        A_{\textrm{ff}}(\vpsi_{x,t}) \defeq B + K_{xu}T_{n_u, n_{s}}(I\kron S\vpsi_{x,t})
    \end{equation}
    \begin{equation}
        \begin{aligned}
        b_{\textrm{ff}}(\vpsi_{x,t}) &\defeq K_{xx}\vpsi_{x,t} + K_BK_ue_t \\&+ K_{xu}(S\vpsi_x \kron K_u e_t) - \vpsi_{x,t+1}^{\textrm{ref}}.
        \end{aligned}
    \end{equation}
\end{subequations}

\subsubsection{CLF-based Tracking under Ellipsoidal Constraints}
To ensure that the ellipsoidal actuation constraints are still satisfied after adding the feedforward term, the optimization becomes a quadratically constrained quadratic program (QCQP):
\begin{equation}\label{eq:clf_opt_ellipsoid}
  \begin{aligned}
        \min_{k_\textrm{ff}} ~&\|k_\textrm{ff}\|^2_2 + \gamma \delta \\
        \subjectto &V(e_{x, t+1}) -\lambda V(e_{x,t}) \leq \delta \\
        & (K_u e_{u,t} + k_{\textrm{ff}})^\top Q_{\vpsi} (K_u e_{u,t} + k_{\textrm{ff}}) < q_{\vpsi}\\
        &\delta \geq 0
  \end{aligned}
\end{equation}

\subsubsection{CLF-based Tracking under Polytopic Constraints}
Ensuring the polytopic constraints, the optimization becomes a quadratic program:
\begin{equation}\label{eq:clf_opt_polytope}
  \begin{aligned}
        \min_{k_\textrm{ff}} ~&\|k_\textrm{ff}\|^2_2 + \gamma \delta \\
        \subjectto &V(e_{x, t+1}) -\lambda V(e_{x,t}) \leq \delta \\
        & G(K_u e_{u,t} + k_{\textrm{ff}}) < g\\
        &\delta \geq 0
  \end{aligned}
\end{equation}

\begin{remark}
Equations \eqref{eq:clf_opt_ellipsoid}, \eqref{eq:clf_opt_polytope} replace the conditional relation in \eqref{eq:set_containment}, $\vpsi_u \in \cP_{\vpsi} \implies u \in \cP_u$ with the hard constraint $u \in \cP_u$; this results in a stronger condition, as the control input is now bounded the entire state space. 
\end{remark}

\begin{remark}
The formulation \eqref{eq:clf_opt_polytope} simply adds a linear constraint, yielding more computationally efficient formulation than its ellipsoid counterpart, which introduces a quadratic constraint. 
\end{remark}

\section{Summary of the Proposed Approach}
Summarizing the previous sections, the overall procedure for constructing the Koopman tracking controller is outlined in Algorithm \ref{alg:koopman_tracking_controller}. The main steps are: selecting a dictionary of observables, collecting a dataset of observables and inputs by exciting the system with a random controller, identifying the Koopman bilinear model, 
synthesizing a stabilizing controller, and computing the feedforward term required for stabilization of a fixed equilibrium or for trajectory tracking.

\begin{algorithm}[h]
\caption{Koopman Tracking Controller Design}
\label{alg:koopman_tracking_controller}
\begin{algorithmic}[1]

\State \textbf{Given:} state constraint set $\mathcal{P}_x$ and input constraint set $\mathcal{P}_u$.

\State \textbf{Select:} observable dictionaries $\vpsi_x(\cdot)$ and $\vpsi_u(\cdot)$ compatible with \eqref{eq:assumption2}.

\State \textbf{Collect data:} generate
$\mathcal{D}=\{(x_t,u_t,x_{t+1})\}_{t=1}^{N}$
using random control inputs, and form the corresponding lifted dataset
$\mathcal{D}_{\vpsi}$.

%\State \textbf{Shift observables:} set $\vpsi_u := \vpsi_x - \vpsi_x^{\textrm{ref}}$, $\vpsi_x(x) := \vpsi_u-\vpsi_u^{\textrm{ref}}$, so that $\bar{\vpsi}_x(x_{\textrm{eq}})=0$.

\State \textbf{Identify bilinear structure:} solve the least-squares problems
associated with \eqref{eq:assumption1} and \eqref{eq:assumption2} to obtain
$H$, $S$, and $C$.

\State \textbf{Identify Koopman bilinear model:} solve the constrained least-squares
problems in \eqref{eq:H_id} and \eqref{eq:koopman_id} to obtain
$K_{xx}$, $K_B$, and $K_{xu}$.

\State \textbf{Synthesize stabilizing controller:}
\If{$\mathcal{P}_u$ is ellipsoidal}
    \If{$\mathcal{P}_{\vpsi}$ is not available}
        \State Construct or approximate $\mathcal{P}_{\vpsi}$ from $\mathcal{P}_x$
        using \eqref{eq:MVE}.
    \State Solve the ellipsoidal constrained synthesis problem in \eqref{eq:controller_ellipsoid_constraints}.
    \EndIf
\ElsIf{$\mathcal{P}_u$ is polytopic}
    \If{$\mathcal{P}_{\vpsi}$ is not available}
        \State solve \eqref{eq:controller_robust_lp}.
    \Else
        \State Use the given lifted set $\mathcal{P}_{\vpsi}$, Solve the polytopic constrained synthesis problem in \eqref{eq:controller_polytopic_constraints}. 
    \EndIf 
\EndIf

\State \textbf{Compute feedforward term:}
\If{a constant setpoint $x^{\mathrm{ref}}$ is given}
    \State Compute the constant feedforward term $k_{\mathrm{ff}}$ using \eqref{eq:u_ref_estimation}.
\ElsIf{a reference trajectory $\{x_t^{\mathrm{ref}}\}$ is given}
    \State Compute the time-varying feedforward term $k_{\mathrm{ff},t}$ by solving the CLF optimization problem in \eqref{eq:clf_opt_ellipsoid} or \eqref{eq:clf_opt_polytope}, depending on the input constraint representation.
\EndIf

\State \textbf{Apply controller:}
\[
    u_t =
    K_u\big(\vpsi_u(x_t)-\vpsi_u(x_t^{\mathrm{ref}})\big)
    +
    k_{\mathrm{ff},t}.
\]

\State \textbf{Output:} Koopman bilinear model, constrained stabilizing controller,
and feedforward tracking controller.
\end{algorithmic}
\end{algorithm}
%First, after selecting the nonlinear observable functions $\vpsi_x$ and $\vpsi_u$, we redefine the observables by shifting the lifted state relative to the reference so that the equilibrium corresponds to the origin in the lifted space. This is done by subtracting the observable evaluated at the equilibrium point,
%\begin{align}
%    \vpsi_x \coloneqq \vpsi_x - \vpsi_x^{\textrm{ref}},\\
%    \vpsi_u \coloneqq \vpsi_u - \vpsi_u^{\textrm{ref}}.    
%\end{align}
%This transformation ensures that the closed-loop bilinear system in \eqref{eq:closed_loop_linear} has an equilibrium at the origin of $\mathbb{R}^{n_{\vpsi_x}}$.
After identifying the model, we compute a stabilizing controller by solving the semidefinite program with polytopic and/or ellipsoidal constraints described in \eqref{eq:controller_polytopic_constraints} and \eqref{eq:controller_ellipsoid_constraints}.

Once the closed loop Koopman model is known, we design feedforward controller term that stabilize the system around a constant setpoint as in \eqref{eq:feedforward_control_law}, or given a reference trajectory, we can compute a time-varying feedforward term using CLF formulation for tracking as in \eqref{eq:clf_opt_ellipsoid} and \eqref{eq:clf_opt_polytope}.

\section{Experiments and Results on Dynamical Systems}\label{sec:experiments}
In this section, we describe the dynamical models, datasets, lifting functions, and constraint sets used to evaluate the proposed approach.
We then report tracking performance and compare the proposed controllers with several Koopman MPC baselines.

We evaluate our method on the following dynamical systems.

\begin{figure*}[t!]
    \centering
    \subfloat[Same initial condition: phase portrait]{%
        \includegraphics[width=0.48\textwidth]{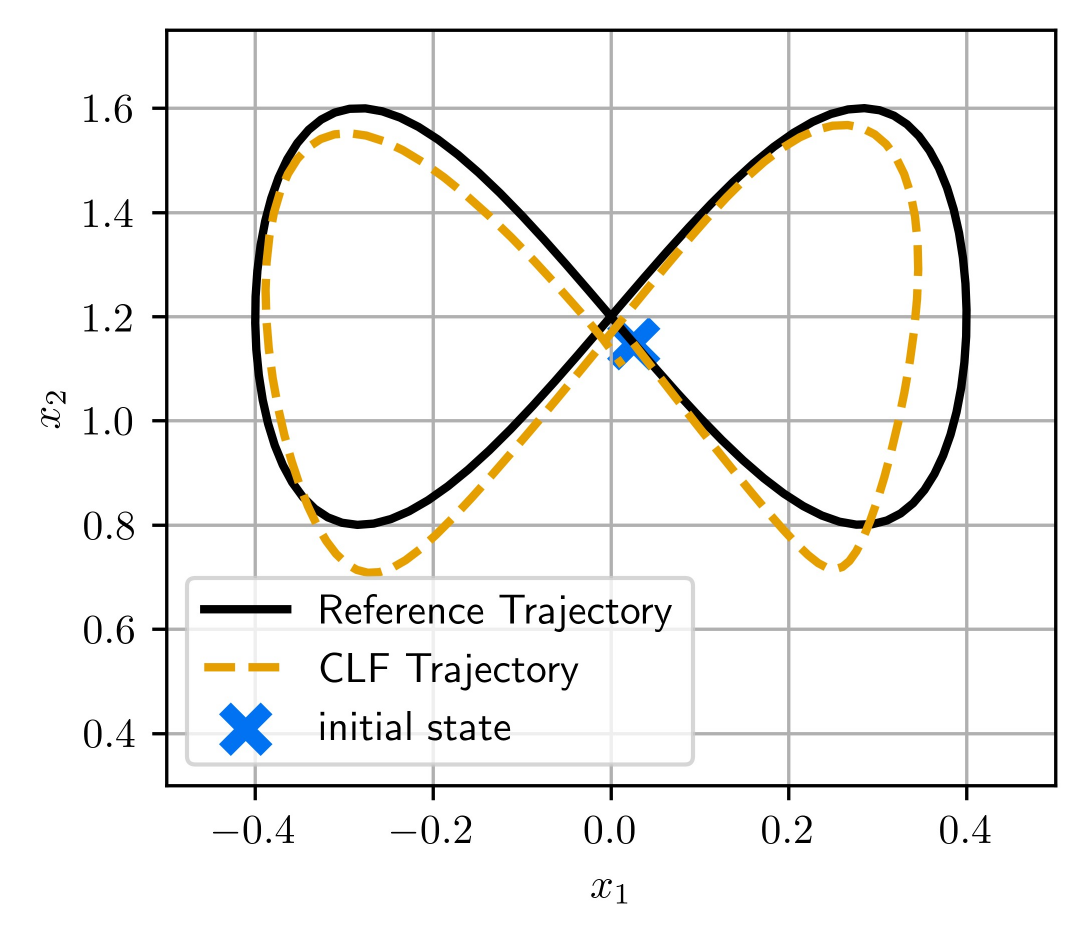}%
        \label{fig:Tracking_double_same_init}
    }\hfill
    \subfloat[Different initial condition: phase portrait]{%
        \includegraphics[width=0.48\textwidth]{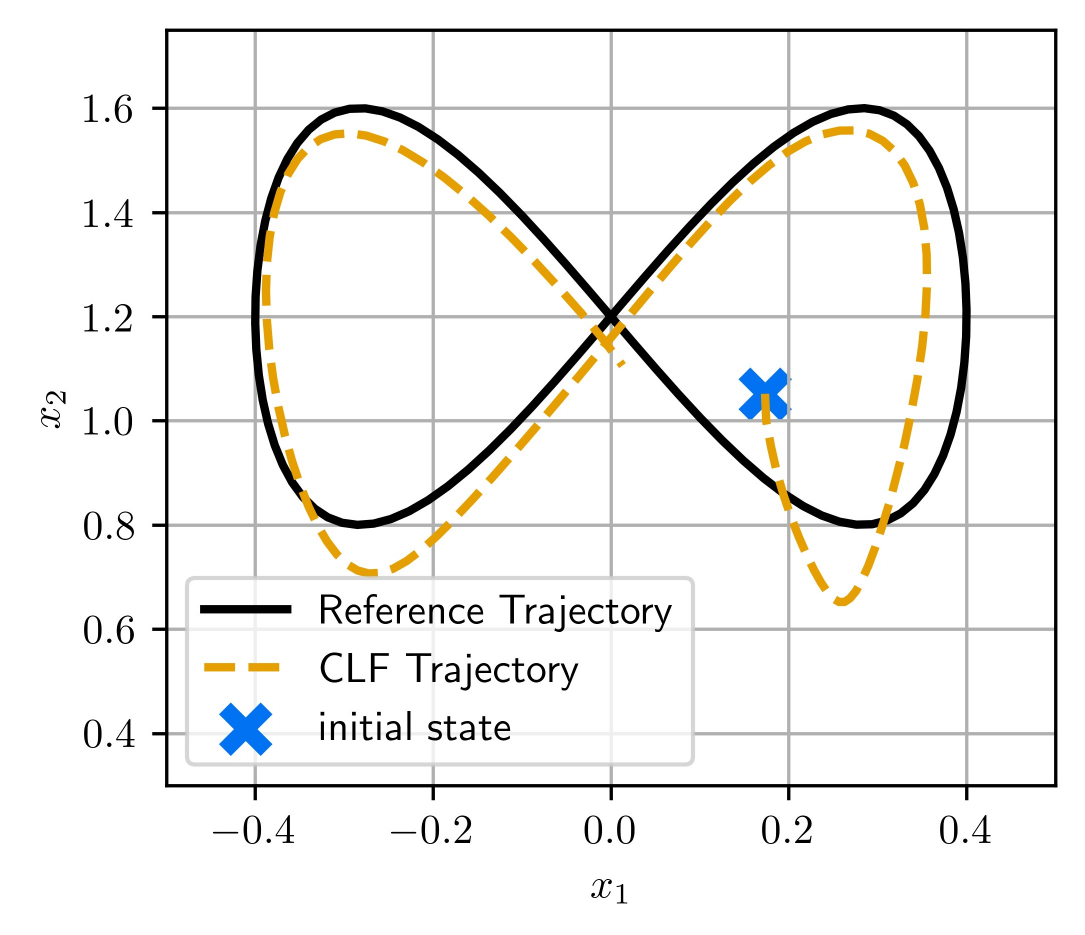}%
        \label{fig:Tracking_double_diff_init}
    }

    \vspace{0.5em}
    \subfloat[Same initial condition: control input]{%
        \includegraphics[width=0.48\textwidth, trim = 0em 0em 0em 3em, clip]
        {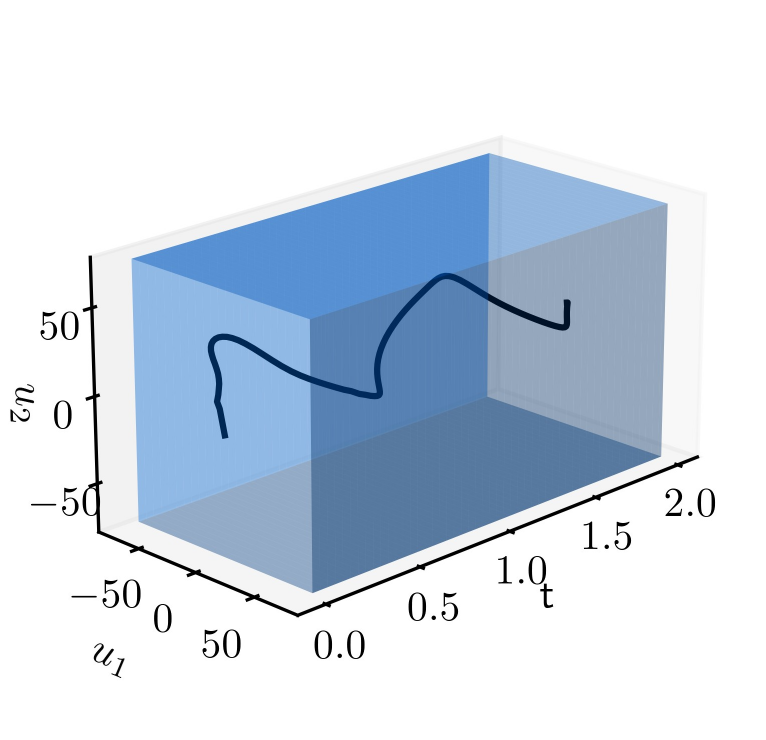}%
        \label{fig:Tracking_double_same_force}
    }\hfill
    \subfloat[Different initial condition: control input]{%
        \includegraphics[width=0.48\textwidth, trim = 0em 0em 0em 3em, clip]{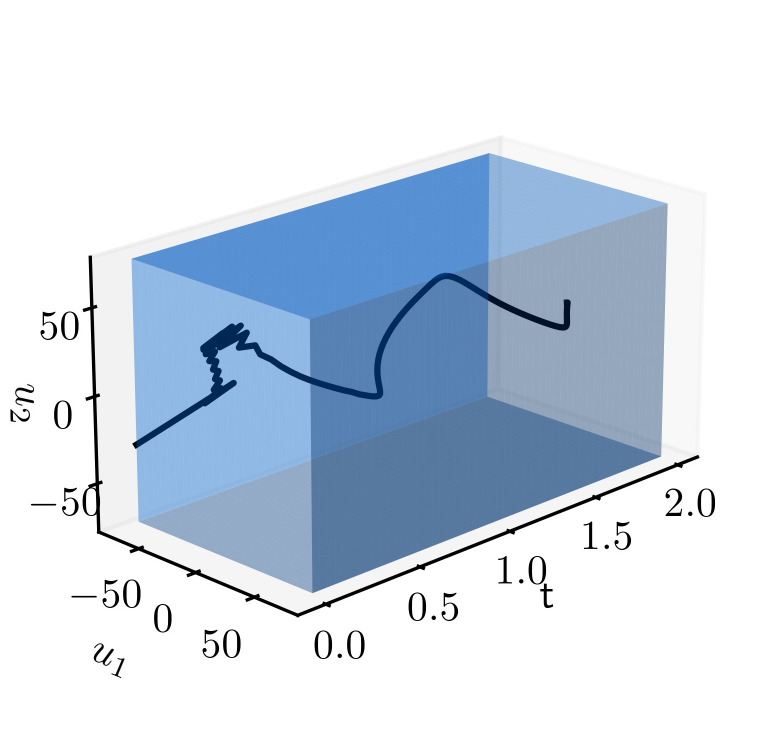}%
        \label{fig:Tracking_double_diff_force}
    }
    \caption{Double pendulum tracking results under the CLF-KoCoP controller. }
    \label{fig:tracking_double}
\end{figure*}

\begin{figure*}[h!]
    \centering
    \subfloat[Same initial condition: phase portrait]{%
        \includegraphics[width=0.48\textwidth]{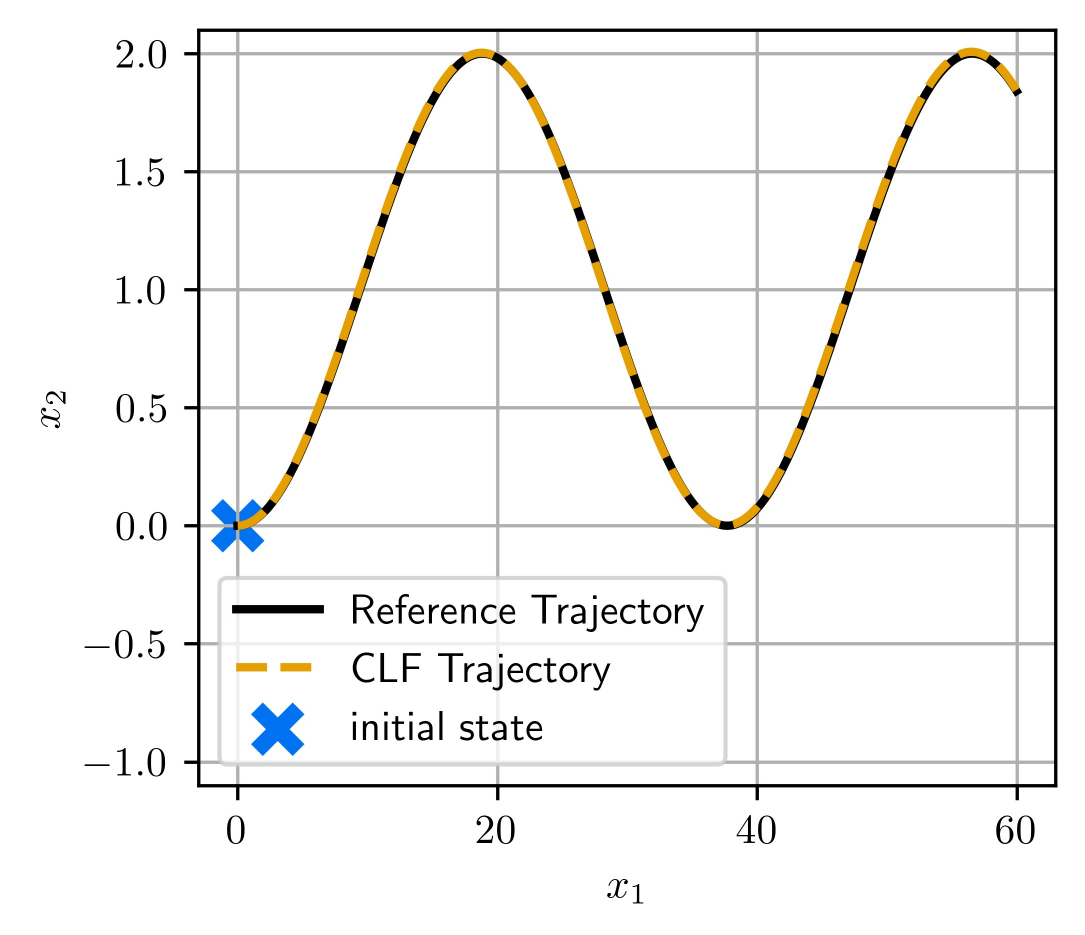}%
        \label{fig:Tracking_uni_same_init}
    }\hfill
    \subfloat[Different initial condition: phase portrait]{%
        \includegraphics[width=0.48\textwidth]{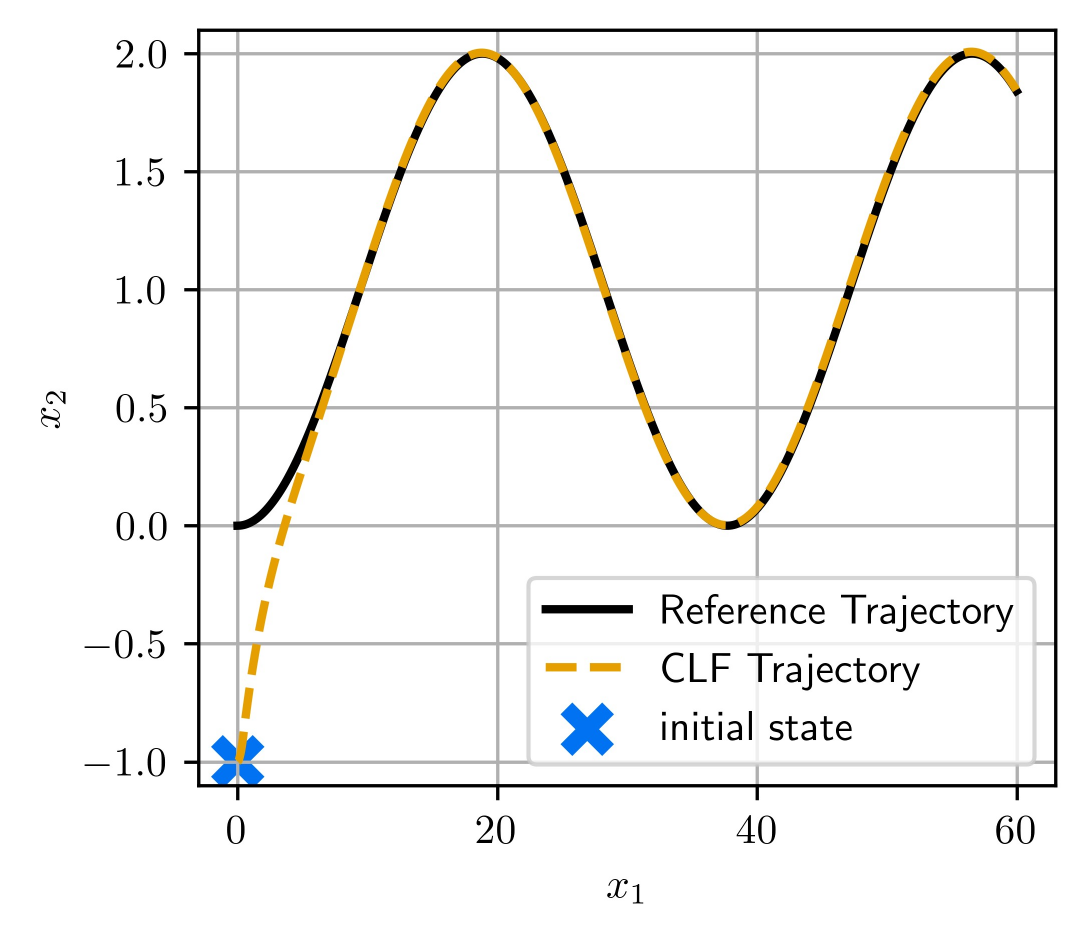}%
        \label{fig:Tracking_uni_diff_init}
    }
    \vspace{0.5em}
    \subfloat[Same initial condition: control input]{%
        \includegraphics[width=0.48\textwidth, trim = 0em 0em 0em 3em, clip]{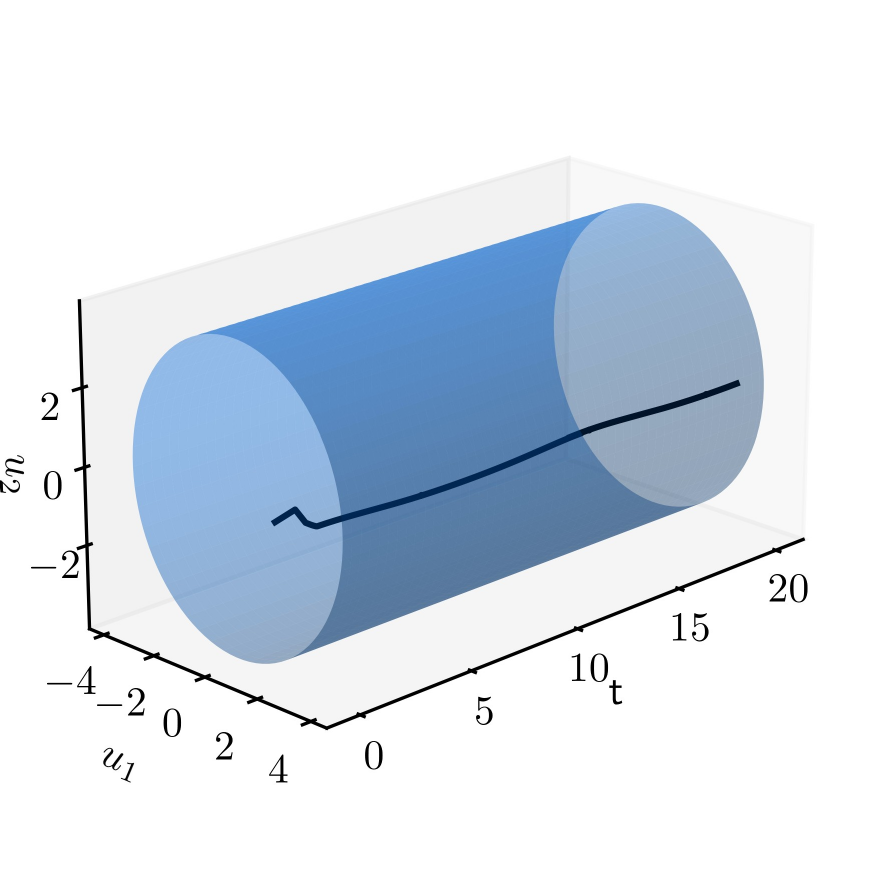}%
        \label{fig:Tracking_uni_same_force}
    }\hfill
    \subfloat[Different initial condition: control input]{%
        \includegraphics[width=0.48\textwidth, trim = 0em 0em 0em 3em, clip]{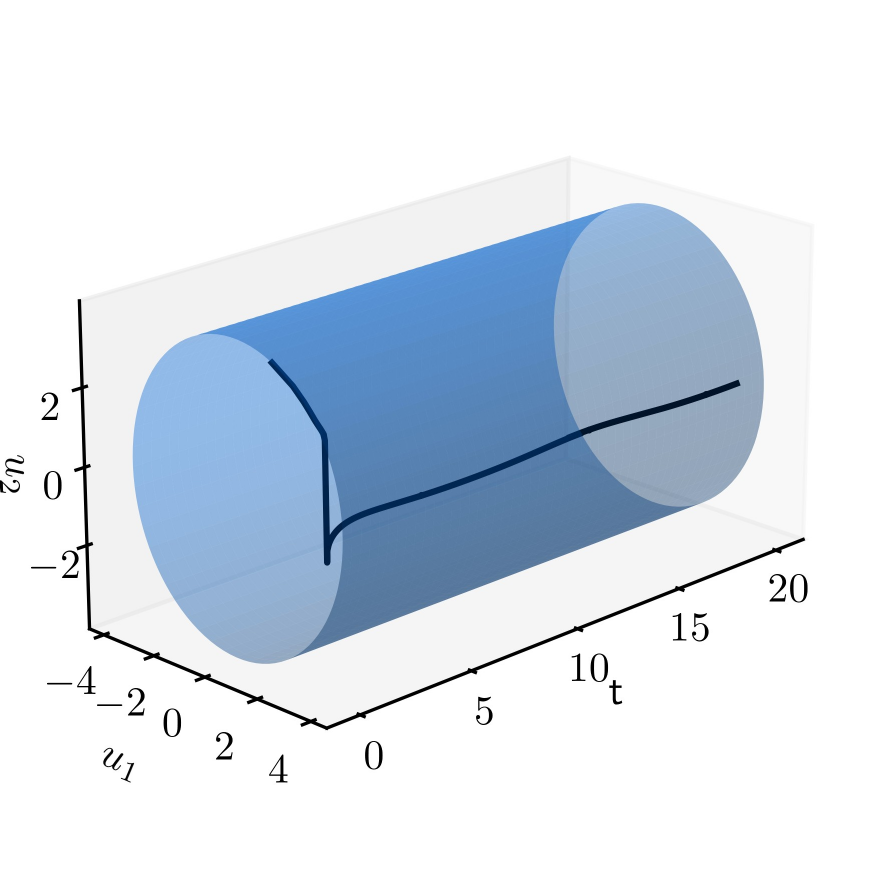}%
        \label{fig:Tracking_uni_diff_force}
    }
    
    \caption{Unicycle tracking results under the CLF-KoCoP controller.}
    \label{fig:tracking_uni}
\end{figure*}

\subsubsection{Inverted Double Pendulum} We denote the state of the double pendulum, consisting of the two link angles and their
angular velocities, as
\begin{equation}
    x = \bmat{\theta_1 & \theta_2 & \dot\theta_1 & \dot\theta_2}^{\top}, \qquad u = \bmat{\tau_1 & \tau_2}^{\top}.
\end{equation}
The dynamics are written as
\begin{equation}
    \dot{x} =
    \bmat{
        \dot\theta_1 \\
        \dot\theta_2 \\
        M(\theta)^{-1}
        \big(u + D(\dot\theta)-C(\theta,\dot\theta)-G(\theta)\big)
    },
    \label{eq:double_pendulum_dynamics}
\end{equation}
where $\theta=\bmat{\theta_1 & \theta_2}^{\top}$,
$\dot\theta=\bmat{\dot\theta_1 & \dot\theta_2}^{\top}$, and
\begin{equation}
    M(\theta) =
    \bmat{
        (m_1+m_2)l_1^2 & m_2l_1l_2\cos\theta_r \\
        m_2l_1l_2\cos\theta_r & m_2l_2^2
    }, \qquad \theta_r:=\theta_2-\theta_1 .
    \label{eq:double_pendulum_mass}
\end{equation}

The remaining terms are

\begin{equation}
\begin{aligned}
    C(\theta,\dot\theta)
    &=
    \bmat{
        -h(2\dot\theta_1\dot\theta_2+\dot\theta_2^2) \\
        h\dot\theta_1^2
    },
    \qquad
    h=m_2l_1l_2\sin\theta_r,\\
    G(\theta)
    &=
    \bmat{
        (m_1+m_2)gl_1\sin\theta_1 \\
        m_2gl_2\sin\theta_2
    },
    \qquad
    D(\dot\theta)
    =
    -b
    \bmat{
        \dot\theta_1 \\
        \dot\theta_2
    }.
\end{aligned}
\label{eq:double_pendulum_terms}
\end{equation}

\noindent While the selection of
observable functions is beyond the scope of this work, we use observables motivated by the analytical structure of the double pendulum dynamics:
\begin{equation}
{\vpsi}_{x}(\vx) =
\begin{bmatrix}
  \vx \\
  \sin (\theta_1) \\
  \cos(2\theta_r) \\
  \dot\theta_1^2 \sin (\theta_r) \\
  \dot\theta_2^2 \sin (\theta_r) \\
  \dot\theta_1^2 \sin (2\theta_r) \\
  \dot\theta_2^2 \sin (2\theta_r) \\
  \sin(\theta_r) \cos(\theta_2) \\
  \sin(2\theta_1 - \theta_2) \\
  \sin (\theta_1 - 2\theta_2)
\end{bmatrix},
{\vpsi}_{u}(\vx) =
\begin{bmatrix}
  \vx \\
  \sin (\theta_1) \\
  \cos (2\theta_r) \\
\end{bmatrix}.
\label{observables_double_pendulum}
\end{equation}
where $\theta_r = \theta_2 - \theta_1$ defined as relative angle between the links, and $\vpsi_x \in \mathbb{R}^{13}, \vpsi_u \in \mathbb{R}^{6}$
capturing the underlying dynamics of the double pendulum through trigonometric terms and their combinations of velocities.

We collect a dataset from temporal snapshots of the dynamics, containing $900$ trajectories, each having length of $100$ samples with the sampling rate of $T_s = 0.02$s. The initial states are sampled uniformly with the range $\theta_{1,2} \in [-5, 5]$ for the link angles and $\dot \theta_{1,2} \in [-\pi, \pi]$ for link velocities. Let the state be $\mathbf{x} = \bmat{\theta_1 & \theta_2 & \dot \theta_1 & \dot \theta_2 }^\top$. 

We define the state-space polytope over which the input constraints are enforced
according to \eqref{eq:set_containment} as the box
\begin{equation}
    A_x =
    \begin{bmatrix}
        I_{4} \\ -I_{4}
    \end{bmatrix},
    \qquad
    b_x = \mathbf{1}_{8},
\end{equation}
and impose the input constraints
\begin{equation}
    G_u =
    \begin{bmatrix}
        I_{2} \\ -I_{2}
    \end{bmatrix},
    \qquad
    h_u = 100\,\mathbf{1}_{4}.
\end{equation}

\subsubsection{Unicycle Model}

Let $[p_x\;p_y\;\theta]^\top$ denote the unicycle state, consisting of position and orientation of the robot. The dynamics are
\begin{equation}
    \bmat{\dot p_x\\ \dot p_y \\ \dot \theta } = \bmat{\cos \theta & 0 \\ \sin \theta & 0 \\ 0 & 1}\bmat{v \\ \omega}.    
\end{equation}

We define the observable functions
\begin{equation}
    \begin{aligned}
    \vpsi_x(x) &=
    \begin{bmatrix}
        p_x & p_y & \theta & \cos\theta & \sin\theta & e_\theta
    \end{bmatrix}^{\top},\\
    \vpsi_u(x) &=
    \begin{bmatrix}
        e_\theta & r
    \end{bmatrix}^{\top},
    \label{eq:unicycle_observables}
    \end{aligned}
\end{equation}
where heading angle error and distance from the center are defined as
\begin{equation}
    e_\theta := \operatorname{atan2}(-y,-x)-\theta,
    \qquad
    r := \sqrt{x^2+y^2}.
    \label{eq:unicycle_aux}
\end{equation}
Overall the lifting uses position, orientation, sinusoidal functions of the orientation, and heading error in the state observables, while the feedback observables use the heading error and distance from the origin.

Using this model, we generate a dataset consisting of $600$ different initial conditions sampled on a grid, and simulate each trajectory for $100$ time steps, yielding a dataset of size $60000$. The sampling time for this system is $T_s = 0.01$.

For the state region, we use the unit ball
\begin{equation}
    Q_x = I_{3 \times 3}, q_x = 1
\end{equation}
and for the input constraint set, we use
\begin{equation}
    Q_u = I_{2 \times 2}, q_u = 16
\end{equation}

The controllers are evaluated on both setpoint regulation and trajectory-tracking tasks, following the formulations in \cref{sec:setpoint,sec:clf_tracking}. We compare the proposed CLF-based method with Koopman MPC baselines in terms of runtime and trajectory mean-squared error relative to the reference. We further evaluate the robustness of the proposed method under state noise and input disturbances.

\subsection{Setpoint Tracking Performance}
To evaluate the setpoint regulation capability of the constrained \kocop{} controller, we stabilize the system around a nonzero equilibrium $\bmat{1 & 1 & 0 & 0}$. The corresponding feedforward term is computed using \eqref{eq:u_ref_estimation}. 

As a comparison, we also apply the CLF-based tracking controller by treating the same
constant setpoint as a reference trajectory. The resulting time responses are compared in
Fig.~\ref{fig:CLFvsErrorDynamics}.

 \begin{figure}[h!]
     \centering \includegraphics[width=1\linewidth]{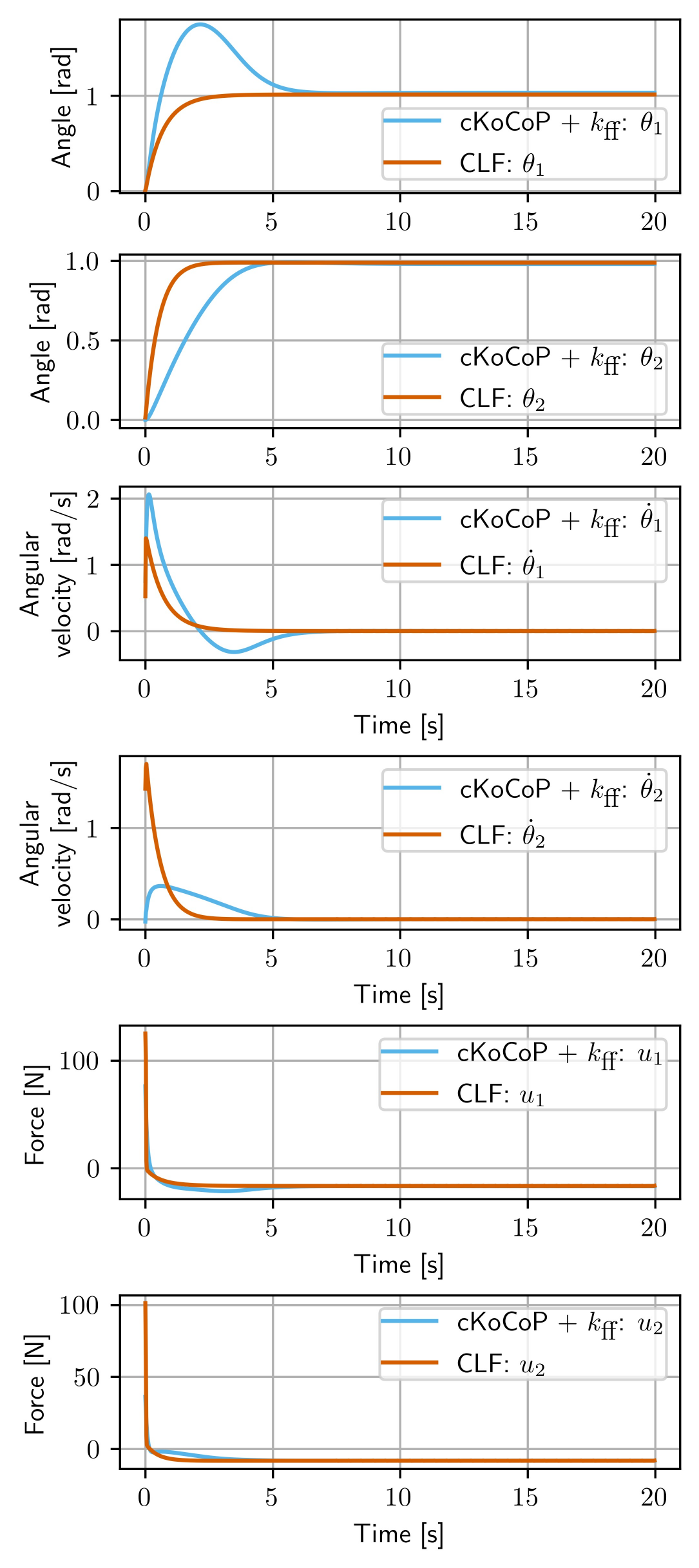}
     \caption{Tracking with CLF compared to tracking with recalculating feedforward term to a static reference.}
     \label{fig:CLFvsErrorDynamics} 
 \end{figure}
 
Both controllers drive the system toward the desired setpoint. The CLF-based controller
achieves a faster transient response because it updates $k_{\textrm{ff}}$ online and
explicitly enforces a one-step Lyapunov decrease. In contrast, the constant
feedforward-feedback controller selects $k_{\text{ff}}$ so that the nominal error
dynamics follow $e_{x,t+1}=\bar K e_{x,t}$.  However, the
constant feedforward-feedback controller is computationally significantly faster: one
control step takes $0.38~ms$ on average, compared with $3.16~ms$ for the CLF-based controller.
\begin{remark}
In the presence of Koopman model mismatch caused by assumptions \eqref{eq:assumption1} and \eqref{eq:assumption2},
this fixed feedforward term can lead to small steady-state errors.
\end{remark}

\subsection{Reference Trajectory Tracking Performance}
To measure the tracking performances, we construct various tracking scenarios for the dynamical systems.
For the double pendulum, we design a task-space reference trajectory with a lemniscate shape, consisting of $500$ steps, corresponding to $10$ seconds. Within this task, we consider two different initial conditions. In the first case, the system starts from the same initial condition as the reference trajectory $\mathbf{x}_0 = \bmat{0.927 & -0.927 & 0& 0}^\top$ which corresponds to the center point of the lemniscate shape with $x=0$, $y=1.2$ in the task space. However, in the second case the system starts from a different initial condition $\mathbf{x}_0 = \bmat{1.127 & -0.827 & 0 & 0}^\top$ and is tasked to follow the same reference trajectory. The corresponding tracking results are demonstrated in Fig.~\ref{fig:Tracking_double_same_init}, for the same initial condition case, and Fig.~\ref{fig:Tracking_double_diff_init} for the different initial condition case. Their associated force plots are presented in Fig.~\ref{fig:Tracking_double_same_force} and Fig.~\ref{fig:Tracking_double_diff_force} 

For the unicycle model, we consider a sinusoidal reference trajectory corresponding to a slalom maneuver. We simulate the system over 2000 time steps, corresponding to 20 seconds. The first trajectory starts from $\bmat{0 & 0 & 0}^\top$, and the second one from $\bmat{0 & -1 & 0}^\top$, corresponding to the same and different initial conditions, respectively. Fig.~\ref{fig:Tracking_uni_same_init} shows the trajectory following the same initial condition, and Fig.~\ref{fig:Tracking_uni_diff_init} displays the different initial condition case. Their corresponding force plots are presented in Fig.~\ref{fig:Tracking_uni_same_force} and Fig.~\ref{fig:Tracking_uni_diff_force}.

Across both systems, the proposed CLF controller tracks the prescribed reference trajectories for both same and different initial conditions as phase-portrait displays. The input plots show that the forces remain within the prescribed admissible sets.

\subsection{Trajectory Tracking Under Noise and Disturbance}
Since the proposed controller is synthesized from a data-driven approach, its closed loop performance may be affected by observation noise, modeling error of the bilinear Koopman system, and unmodeled actuation disturbances. We therefore evaluate the robustness of the proposed controller under noisy reference observations and input disturbances.
\begin{equation}
    \hat{x}_{\textrm{ref}} = x_{\textrm{ref}} + \epsilon_x,
\end{equation}
and model input disturbances as
\begin{equation}
    x_{t+1} = f(x_t, u_t + \epsilon_u).
\end{equation}

We assume $\epsilon_x, \epsilon_u \sim \mathcal{N}(0, \sigma^2 I)$ for the double pendulum system. We vary the noise level over $\sigma^2 \in \{0.2, 0.4, \dots, 1.0\}$ and evaluate performance by computing steady-state errors. For each noise level, results are averaged over 10 independent trajectories of length $N = 100$, using a step reference input.

A representative trajectory under the highest noise level is shown in Fig.~\ref{fig:noisy_pendulum}, which displays the end-effector phase portrait. Fig.~\ref{fig:noisyPendulumFeatures} shows the corresponding angle and angular-velocity tracking errors for the same run.

Despite significant perturbations, the controller is able to track the reference. The steady-state errors and MSEs of the trajectories are displayed with respect to different noise variances in Fig.~\ref{fig:noisy_error}. As expected, higher noise levels lead to increased steady-state error, but the degradation remains gradual, indicating robustness of the proposed method.

\begin{figure}[h!]
    \centering
    \includegraphics[width=1\linewidth]{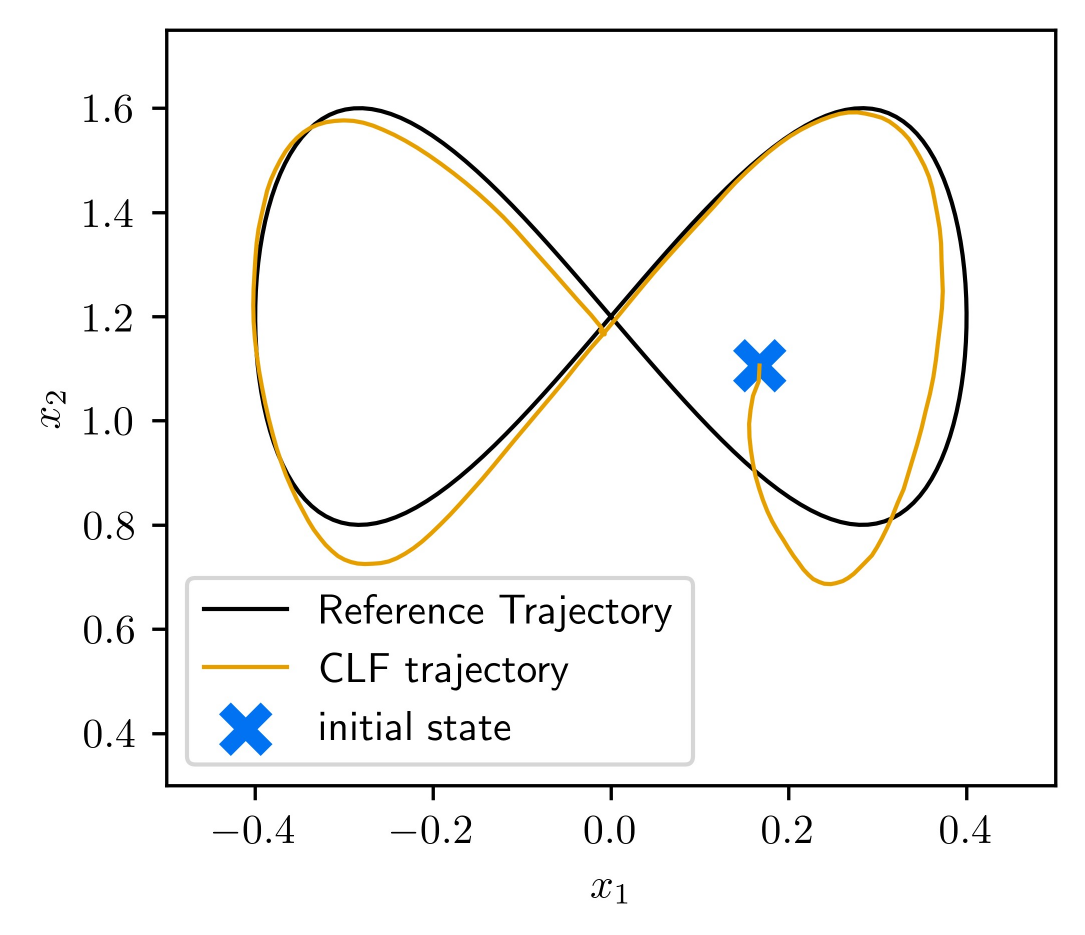}
    \caption{Representative trajectory of the pendulum under observation and input noise with $\epsilon_x, \epsilon_u \sim \mathcal{N}(0, I)$. The controller maintains stable tracking despite significant noise.}
    \label{fig:noisy_pendulum}
\end{figure}

\begin{figure}[h!]
    \centering
    \includegraphics[width=1\linewidth]{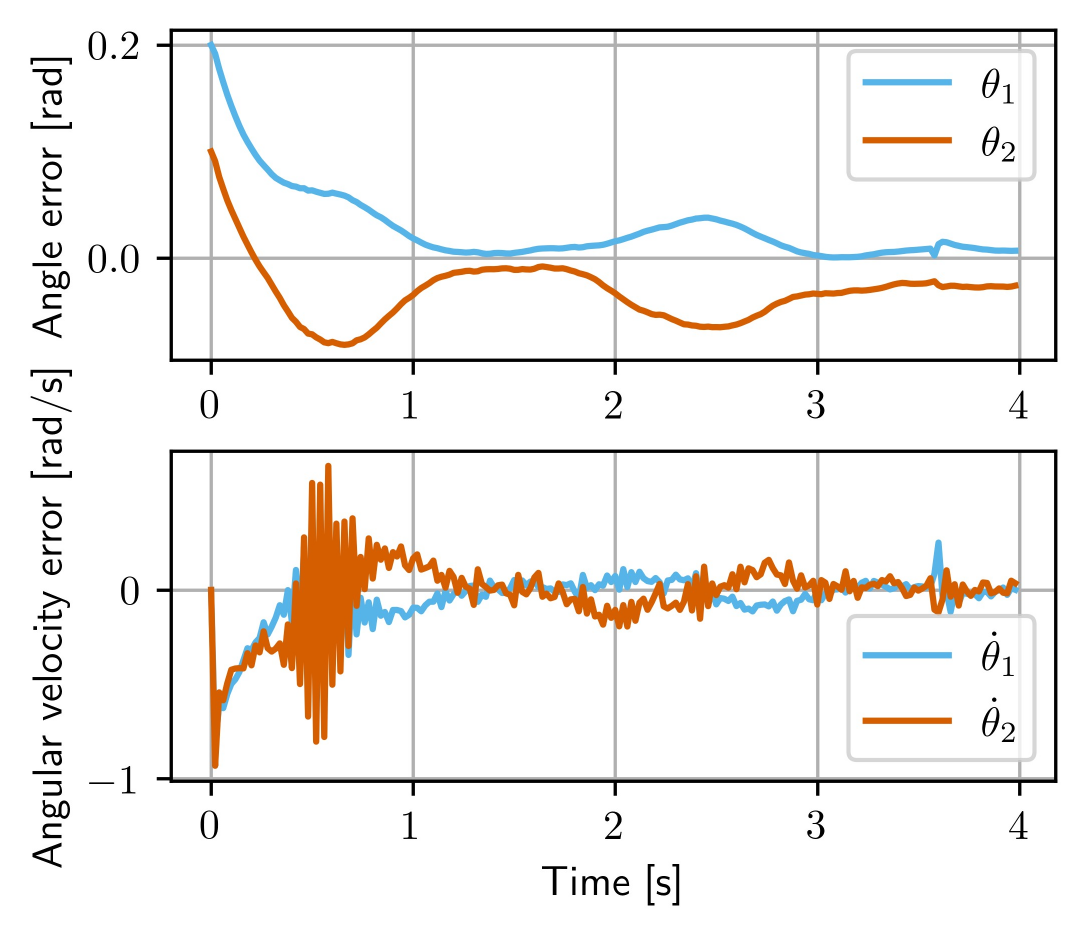}
    \caption{Evolution of angle and angular velocity under noise of $\epsilon_x, \epsilon_u \sim \mathcal{N}(0, 0.5^2 I)$.}
    \label{fig:noisyPendulumFeatures}
\end{figure}

\begin{figure}[h!]
    \centering
    \includegraphics[width=1\linewidth]{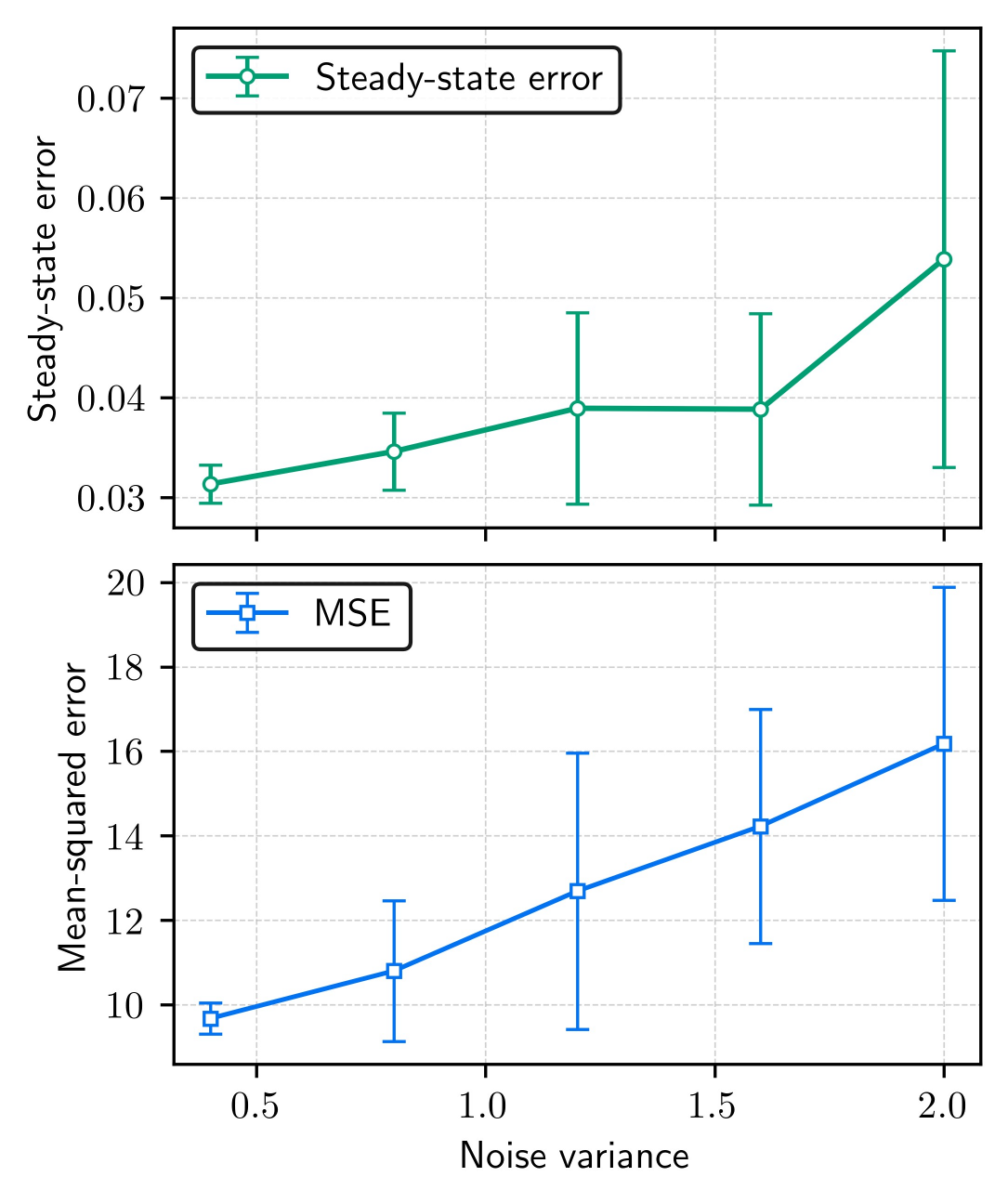}
    \caption{Steady-state tracking error as a function of noise level. Results are averaged over 10 trajectories of length $N=100$ for each $\sigma^2$.}
    \label{fig:noisy_error}
\end{figure}

\subsection{Comparison of Computational Runtimes}
We first report the average cost of one CLF optimization step. A single CLF optimization for a double pendulum with polytopic constraints requires 3.58 ms on average. In comparison, the CLF formulation with ellipsoidal constraints on the unicycle model requires 3.86 ms per optimization on average. Although the double pendulum model has larger observable dimensions $n_{\vpsi_x}$ and $n_{\vpsi_u}$ than the unicycle model, it remains slightly faster because the polytopic formulation introduces linear constraints, whereas the ellipsoidal formulation leads to a quadratic constraint. All computations were performed on a machine equipped with an Intel Core i7-10750 CPU at 2.60 GHz with 6 cores.

We further compare our method with three different MPC approaches similar to the one presented in \cite{bruder2021advantages}, K-LMPC denotes that the model is linearized around the origin. K-BMPC on the other hand, utilizes the bilinear structure to linearize the system around the initial state, yielding the approximate transition
\begin{equation}
    \vpsi_{x,t+1} = K_{xx}\vpsi_x + Bu_t + K_{xu}\vpsi_{x, t=0}.
\end{equation}
We further model a nonlinear MPC (K-NMPC) utilizing the CasADi software package \cite{Andersson2019} that utilizes the bilinear Koopman structure. Finally, we model MPC in true nonlinear state-space dynamics of the system as a baseline. For stage costs, and terminal cost, we re-use the Lyapunov matrix $P$ to ensure a fair comparison. We set control costs as $R = 10^{-4}$ for both dynamical systems, for each of the algorithms K-LMPC, K-BMPC and K-NMPC.

We sample four lemniscate-shaped trajectories with different sizes, and sample initial conditions uniformly around $\bmat{0.25& 0.25& 0& 0}$ for the double pendulum, and sample four "slalom" shapes with different amplitudes and frequencies for the unicycle. We compute mean square errors (MSE) and runtimes of a single step optimization of both algorithms. For MPCs we compute the time horizons of $T = 10, 20$ and $30$, and report the results in Table \ref{tab:dp_mse_runtime} for double pendulum and Table \ref{tab:uni_mse_runtime} for the unicycle models. We observe that the proposed CLF method performs the fastest runtimes for solving the control problem in both of the systems, while having comparable MSE performances. Since the observable space is larger in double pendulum, the runtime for CLF ($3.58~ms$) is significantly lower than the fastest MPC baseline, K-LMPC ($20.10~ms$). This gap is expected because the MPC methods solve a finite-horizon optimization problem whose computational cost grows with the observable dimension, whereas the CLF problem solves a single step optimization.

\begin{table}[h]
    \centering
    \caption{Double-pendulum tracking performance and computation time. Values are reported as mean $\pm$ standard deviation.}
    \label{tab:dp_mse_runtime}
    \begin{tabular}{lcc}
        \toprule
        Controller & MSE & Runtime (ms) \\
        \midrule
        CLF           & $0.009675 \pm 0.002723$ & $\mathbf{3.26 \pm 1.34}$ \\
        K-LMPC $T=10$ & $0.008536 \pm 0.002133$ & $8.06 \pm 3.97$ \\
        K-LMPC $T=20$ & $0.008479 \pm 0.002132$ & $17.56 \pm 8.74$ \\
        K-LMPC $T=30$ & $0.008446 \pm 0.002145$ & $24.46 \pm 16.14$ \\
        K-BMPC $T=10$ & $0.008536 \pm 0.002133$ & $8.41 \pm 3.47$ \\
        K-BMPC $T=20$ & $0.008479 \pm 0.002132$ & $17.77 \pm 8.70$ \\
        K-BMPC $T=30$ & $0.008446 \pm 0.002145$ & $24.62 \pm 16.20$ \\
        K-NMPC $T=10$ & $0.008643 \pm 0.002826$ & $46.24 \pm 5.20$ \\
        K-NMPC $T=20$ & $0.008549 \pm 0.002879$ & $85.27 \pm 8.26$ \\
        K-NMPC $T=30$ & $0.008641 \pm 0.002871$ & $126.36 \pm 11.04$ \\
        NMPC $T=10$    & $0.002426 \pm 0.001360$ & $32.23 \pm 3.45$ \\
        NMPC $T=20$    & $0.002767 \pm 0.001432$ & $58.45 \pm 5.55$ \\
        NMPC $T=30$    & $0.002653 \pm 0.001398$ & $85.91 \pm 8.71$ \\
        \bottomrule
    \end{tabular}
\end{table}

\begin{table}[h]
    \centering
    \caption{Unicycle tracking performance and computation time. Values are reported as mean $\pm$ standard deviation.}
    \label{tab:uni_mse_runtime}
    \begin{tabular}{lcc}
        \toprule
        Controller & MSE & Runtime (ms) \\
        \midrule
        CLF           & $0.004396 \pm 0.000123$ & $\mathbf{3.24 \pm 0.71}$ \\
        K-LMPC $T=10$ & $0.003076 \pm 0.001931$ & $6.84 \pm 3.23$ \\
        K-LMPC $T=20$ & $0.003439 \pm 0.001916$ & $12.08 \pm 6.81$ \\
        K-LMPC $T=30$ & $0.004106 \pm 0.002368$ & $17.02 \pm 13.21$ \\
        K-BMPC $T=10$ & $0.003076 \pm 0.001931$ & $7.16 \pm 3.39$ \\
        K-BMPC $T=20$ & $0.003439 \pm 0.001916$ & $12.46 \pm 6.75$ \\
        K-BMPC $T=30$ & $0.004106 \pm 0.002368$ & $17.05 \pm 8.64$ \\
        K-NMPC $T=10$ & $0.002380 \pm 0.000493$ & $33.34 \pm 4.30$ \\
        K-NMPC $T=20$ & $0.002641 \pm 0.000742$ & $60.85 \pm 8.28$ \\
        K-NMPC $T=30$ & $0.004751 \pm 0.001603$ & $88.83 \pm 11.27$ \\
        NMPC $T=10$    & $0.000383 \pm 0.000074$ & $15.12 \pm 2.42$ \\
        NMPC $T=20$    & $0.000289 \pm 0.000078$ & $26.56 \pm 3.81$ \\
        NMPC $T=30$    & $0.000280 \pm 0.000078$ & $37.55 \pm 5.31$ \\
        \bottomrule
    \end{tabular}
\end{table}

\begin{figure}[h!]
    \centering
    \includegraphics[width=1\linewidth]{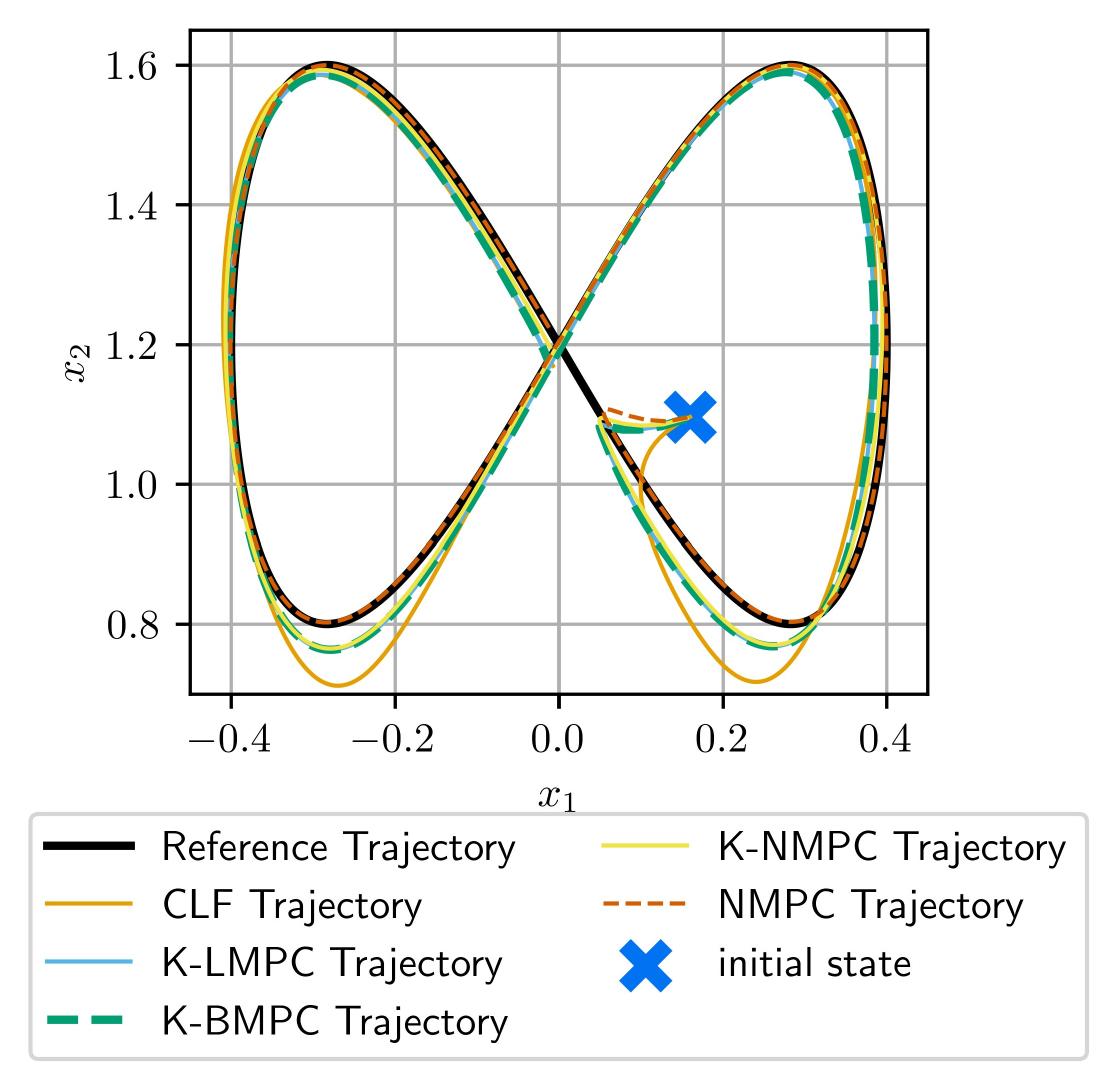}
    \caption{Comparison of CLF with different MPC types}
    \label{fig:trajectory_comparison}
\end{figure}

\section{Conclusion}
This paper presented a method for stabilizing and tracking reference trajectories under input constraints. An optimization problem is presented that follows the actuation limits that are modeled to be polytopic or ellipsoid sets. Under the same constraints, we expand the results to stabilizing the system around a provided setpoint with a feedforward term, and finally, we have presented a CLF-based tracking algorithm to extend it to track a time varying reference trajectory. Since the provided algorithm does not directly address the modeling errors between the Koopman system and the actual dynamical system, future work will investigate the use of integral-based control terms to compensate for steady-state errors. 

\section*{APPENDIX}\label{sec:appendix}
We provide additional discussion for validity of Lyapunov function under the presence of modeling uncertainties generated by the assumptions, and provide an extended sufficient condition for validation the Lyapunov function.

\subsection{Validity of the Lyapunov Condition under Model Mismatch}
\label{app:lyapunov_mismatch}

We consider the effect of mismatch in the structural relation
\eqref{eq:assumption1}. In particular, suppose that the identity
$S\vpsi_x\kron \vpsi_u = H\vpsi_x$ holds only approximately, i.e.,
\begin{equation}
    S\vpsi_x\kron \vpsi_u
    =
    H\vpsi_x+\epsilon,
    \label{eq:approx_assumption1}
\end{equation}
where the mismatch $\epsilon$ is bounded relative to the lifted state as
\begin{equation}
    \|\epsilon\|_2 \leq \gamma \|\vpsi_x\|_2.
    \label{eq:epsilon_bound}
\end{equation}
Under the feedback law $u=K_u\vpsi_u$, the true closed-loop lifted dynamics become
\begin{equation}
    \vpsi_{x,t+1}
    =
    \tilde K\vpsi_{x,t}
    +
    D\epsilon_t,
    \qquad
    D:=K_{xu}(I\kron K_u),
    \label{eq:mismatched_closed_loop}
\end{equation}
where $\tilde K$, defined as in \eqref{eq:closed_loop_linear_matrix}
is the nominal closed-loop matrix obtained under the exact structural assumption.

We seek a quadratic Lyapunov function
\begin{equation}
    V(\vpsi_x)=\vpsi_x^\top P\vpsi_x,
    \qquad P\succ0,
\end{equation}
that remains decreasing for all admissible mismatch signals satisfying
\eqref{eq:epsilon_bound}. Specifically, we require
\begin{equation}
    V(\vpsi_{x,t+1})-\lambda V(\vpsi_{x,t})\leq0,
    \qquad \lambda\in[0,1),
    \label{eq:robust_lyap_requirement}
\end{equation}
for all pairs $(\vpsi_{x,t},\epsilon_t)$ such that
$\|\epsilon_t\|_2^2\leq \gamma^2\|\vpsi_{x,t}\|_2^2$.

Substituting \eqref{eq:mismatched_closed_loop} into
\eqref{eq:robust_lyap_requirement} gives
\begin{equation}
\begin{aligned}
    &
    (\tilde K\vpsi_x+D\epsilon)^\top
    P
    (\tilde K\vpsi_x+D\epsilon)
    -
    \lambda \vpsi_x^\top P\vpsi_x
    \leq 0 .
\end{aligned}
\label{eq:robust_lyap_quadratic}
\end{equation}
Equivalently,
\begin{equation}
    \begin{bmatrix}
        \vpsi_x \\ \epsilon
    \end{bmatrix}^{\top}
    \begin{bmatrix}
        \tilde K^\top P\tilde K-\lambda P & \tilde K^\top P D \\
        D^\top P\tilde K & D^\top P D
    \end{bmatrix}
    \begin{bmatrix}
        \vpsi_x \\ \epsilon
    \end{bmatrix}
    \leq0 .
    \label{eq:robust_lyap_matrix_form}
\end{equation}
The mismatch bound \eqref{eq:epsilon_bound} can be written as
\begin{equation}
    \begin{bmatrix}
        \vpsi_x \\ \epsilon
    \end{bmatrix}^{\top}
    \begin{bmatrix}
        \gamma^2 I & 0 \\
        0 & -I
    \end{bmatrix}
    \begin{bmatrix}
        \vpsi_x \\ \epsilon
    \end{bmatrix}
    \geq 0 .
    \label{eq:mismatch_quadratic_constraint}
\end{equation}
By the S-procedure, a sufficient condition for
\eqref{eq:robust_lyap_quadratic} to hold for all
$\epsilon$ satisfying \eqref{eq:epsilon_bound} is the existence of a scalar
$\tau\geq0$ such that
\begin{equation}
    \begin{bmatrix}
        \tilde K^\top P\tilde K-\lambda P & \tilde K^\top P D \\
        D^\top P\tilde K & D^\top P D
    \end{bmatrix}
    +
    \tau
    \begin{bmatrix}
        \gamma^2 I & 0 \\
        0 & -I
    \end{bmatrix}
    \preceq 0 .
    \label{eq:robust_lyap_lmi}
\end{equation}
Equivalently,
\begin{equation}
    \begin{bmatrix}
        \tilde K^\top P\tilde K-\lambda P+\tau\gamma^2 I
        &
        \tilde K^\top P D \\
        D^\top P\tilde K
        &
        D^\top P D-\tau I
    \end{bmatrix}
    \preceq 0
    \label{eq:robust_lyap_lmi_expanded}
\end{equation}
with $P\succ0$, $\lambda\in[0,1)$, and $\tau\geq0$. Therefore, if
\eqref{eq:robust_lyap_lmi_expanded} is feasible, then $P$ remains a valid Lyapunov
matrix for mismatched closed-loop dynamics \eqref{eq:mismatched_closed_loop}.

\begin{remark}
    The bound $\gamma$ can be computed as the maximum over the dataset
    \begin{equation}
        \gamma = \max_{(\vpsi_x, \vpsi_u) \in \cD} \frac{\|\epsilon\|}{\|\vpsi_x\|} = \max_{(\vpsi_x, \vpsi_u) \in \cD} \frac{\|\vpsi_x - S\vpsi_x \kron \vpsi_u\|}{\|\vpsi_x\|}
    \end{equation}
\end{remark}

\bibliographystyle{biblio/IEEEtran}
% \bibliography{biblio/IEEEAbrv,biblio/IEEEConfAbrv,biblio/OtherAbrv,biblio/control,biblio/controlAdaptive,biblio/controlCBFs,biblio/koopman,biblio/koopmanControl,biblio/koopmanSoftRobots,biblio/koopmanFeedbackEstimation,biblio/math,biblio/optimization,biblio/tron,biblio/writing,biblio/education,biblio/formationControl,biblio/websites}
\bibliography{biblio/IEEEFull,biblio/IEEEConfFull,biblio/OtherFull,biblio/control,biblio/controlAdaptive,biblio/controlCBFs,biblio/koopman,biblio/koopmanControl,biblio/koopmanSoftRobots,biblio/koopmanFeedbackEstimation,biblio/math,biblio/optimization,biblio/tron,biblio/writing,biblio/education,biblio/formationControl,biblio/websites}

\end{document}

%% file: preamble/common.tex
\input{preamble/commonNoTikz}
\input{preamble/graphicsTikz}
\input{preamble/robotics}

%% file: preamble/commonNoTikz.tex
\input{preamble/fixes}
\input{preamble/graphics}
\input{preamble/pagination}
\input{preamble/math}
\input{preamble/operators}
\input{preamble/notation}
\input{preamble/units}
\input{preamble/markupAndCommenting}
\input{preamble/utilities}

%% file: preamble/fixes.tex
\usepackage{fix-cm}
\usepackage{etex}

\usepackage{dblfloatfix}

\usepackage{nag}

%% file: preamble/graphics.tex
\makeatletter
\@ifpackageloaded{xcolor}{}{%
\usepackage[table,x11names,dvipsnames,svgnames]{xcolor}%
}
\makeatother

\usepackage{colortbl}

\usepackage{graphicx}
\usepackage{wrapfig}

\input{preamble/graphicsColors}

%% file: preamble/graphicsColors.tex
\definecolorset{RGB}{lyft}{}{Red,194,39,36;Sunset,202,53,33;Orange,205,68,20;Amber,200,117,42;Yellow,242,169,52;Citron,186,188,44;Lime,112,159,33;Green,56,139,31;Mint,45,118,56;Teal,52,133,135;Cyan,60,132,202;Blue,55,94,248;Indigo,64,13,247;Purple,115,42,248;Pink,176,25,145;Rose,176,32,75}

\definecolorset{HTML}{h}{}{grapefruit,ED5565;bittersweet,FC6E51;sunflower,FFCE54;grass,A0D468;mint,48CFAD;aqua,4FC1E9;bluejeans,5D9CEC;lavender,AC92EC;pinkrose,EC87C0;lightgray,F5F7FA;mediumgray,CCD1D9;darkgray,656D78}

%% file: preamble/pagination.tex
\usepackage{cite}

\usepackage{microtype}

\usepackage[american]{babel}

\usepackage{array}
\usepackage{multirow}
\usepackage{booktabs}
\usepackage{makecell} %introduces \thead and \makecell

\ifcsname labelindent\endcsname
\let\labelindent\relax
\fi
\usepackage[inline]{enumitem}

\usepackage{subfig}

\newenvironment{lenumerate}[2][]
{\begin{enumerate}[label=(#2\arabic*),leftmargin=0.2in,itemindent=0.15in,#1]}
{\end{enumerate}}

\setlist*[enumerate,1]{label={\itshape\arabic*)}}

\makeatletter
\newcommand{\paragraphswithstop}{%
\let\copyparagraph\paragraph%
\renewcommand\paragraph[1]{\copyparagraph{##1.}}%
}
\makeatother

\usepackage[framemethod=tikz]{mdframed}

\makeatletter
\def\namedlabel#1#2{\begingroup
  #2%
  \def\@currentlabel{#2}%
  \phantomsection\label{#1}\endgroup
}
\makeatother

\makeatletter
\def\namedlabelphantom#1#2{\begingroup
  \def\@currentlabel{#2}%
  \phantomsection\label{#1}\endgroup
}
\makeatother

\newcommand{\parunskip}{\bgroup\unskip\parfillskip=0pt \par\egroup}

%% file: preamble/operators.tex
\newcommand{\bmat}[1]{\begin{bmatrix}#1\end{bmatrix}}

\newcommand{\transpose}{^\mathrm{T}}

\newcommand{\defeq}{\doteq}

\newcommand{\vct}[1]{\mathbf{#1}}

\DeclareMathOperator*{\argmax}{\arg\!\max}

\newcommand{\kron}{\otimes}

\newcommand{\subjectto}{\textrm{subject to}\;}

%% file: preamble/notation.tex
\providecommand{\vx}{\vct{x}}

\providecommand{\vpsi}{\bm{\psi}}

\providecommand{\cD}{\mathcal{D}}

\providecommand{\cP}{\mathcal{P}}

\providecommand{\cS}{\mathcal{S}}

%% file: preamble/units.tex
\usepackage{units}

%% file: preamble/markupAndCommenting.tex
  \newcommand{\newcolorlabel}[2]{%
  \expandafter\newcommand\csname #1\endcsname[1]{%
    \tikz[baseline]{\node[text=white,fill=#2,anchor=base,text height=1.3ex,text depth=0.1ex,font=\sffamily\bfseries]{##1}}}%
}

\newcommand{\newcommenter}[2]{%
  \expandafter\newcommand\csname #1\endcsname[1]{%
    \fcolorbox{#2}{#2}{\color{white}\textsf{\textbf{#1}}}
    {\color{#2}##1}}%
  \expandafter\newcommand\csname at#1\endcsname{%
    \fcolorbox{#2}{#2}{\color{white}\textsf{\textbf{@#1}}}
    {\color{#2}}}%
  \expandafter\newcommand\csname #1cite\endcsname[1]{%
    \csname #1\endcsname {[##1]}
  }%
  \expandafter\newcommand\csname #1ref\endcsname[1]{%
    \csname #1\endcsname {$\blacktriangleright$##1}
  }%
  \expandafter\newcommand\csname #1hl\endcsname[2]{%
    \colorbox{#2}{\color{white}\textsf{\textbf{#1}}}\sethlcolor{Azure2}\hl{##2}~%
    \expandafter\ifx\csname commentarrow\endcsname\relax$\leftarrow$\else \commentarrow[#2]\fi~%
    {\color{#2}##1}}%
  \expandafter\newcommand\csname #1st\endcsname[2]{%
    \colorbox{#2}{\color{white}\textsf{\textbf{#1}}}\sout{##2}~%
    \expandafter\ifx\csname commentarrow\endcsname\relax$\leftarrow$\else \commentarrow[#2]\fi~%
    {\color{#2}##1}}%
}
\newcommenter{TODO}{DodgerBlue1}
\newcommenter{rtron}{Green3}
\newcommenter{taha}{Firebrick1}
\newcommenter{AltConvOpt}{RoyalBlue3}

\usepackage{comment}

\usepackage{pdfcomment}

\usepackage{soul}

\usepackage[normalem]{ulem}

\usepackage{csquotes}

%% file: preamble/utilities.tex
\usepackage{suffix}

\usepackage{environ}

\makeatletter
\newsavebox{\boxifnotempty}
\newcommand{\displayifnotempty}[3]{\sbox\boxifnotempty{#2}\setbox0=\hbox{\usebox{\boxifnotempty}\unskip}%
  \ifdim\wd0=0pt
  \else
  #1\usebox{\boxifnotempty}#3%
  \fi%
}

\newcommand{\ifempty}[2]{\setbox0=\hbox{#1\unskip}%
  \ifdim\wd0=0pt%
  #2%
  \fi%
}

\newcommand{\ifnotempty}[2]{\setbox0=\hbox{#1\unskip}%
  \ifdim\wd0>0pt%
  #2%
  \fi%
}
\makeatother

\newcommand{\switchifempty}[3]{\sbox\boxifnotempty{#1}\setbox0=\hbox{\usebox{\boxifnotempty}\unskip}%
  \ifdim\wd0=0pt{}%
  #2%
  \else{}%
  #3%
  \usebox{\boxifnotempty}%
  \fi%
}

\makeatletter
\@ifundefined{chapter}{\usepackage{algorithm}}{\usepackage[chapter]{algorithm}}
\makeatother
\usepackage{algorithmicx}
\usepackage{algpseudocode}
\makeatother%

\usepackage{scrlfile}

\makeatletter
\newcommand*\newstoreddef[1]{
  \BeforeClosingMainAux{%
    \immediate\write\@auxout{%
      \string\restoredef{#1}{\csname #1\endcsname}%
    }%
  }%
}
\newcommand*{\restoredef}[2]{% used at the aux file
  \expandafter\gdef\csname stored@#1\endcsname{#2}%
}
\newcommand*{\storeddef}[1]{
  \@ifundefined{stored@#1}{0}{\csname stored@#1\endcsname}%
}
\makeatother

\usepackage{pageslts}
\NewEnviron{tee}{\BODY\typeout{Marker Tee [start] ^^J \BODY ^^JMaker Tee [end]}}

\usepackage{cleveref}
\usepackage{nameref}

%% file: preamble/graphicsTikz.tex
\usepackage{tikz}
\usetikzlibrary{calc}
\usetikzlibrary{matrix}
\usetikzlibrary{chains,scopes}
\usetikzlibrary{shapes.geometric}
\usetikzlibrary{arrows.meta}
\usetikzlibrary{decorations.markings}
\usetikzlibrary{decorations.pathreplacing}
\usetikzlibrary{backgrounds}

\tikzset{
  dim above/.style={to path={\pgfextra{
        \pgfinterruptpath
        \draw[>=latex,|->|] let
        \p1=($(\tikztostart)!1.5em!90:(\tikztotarget)$),
        \p2=($(\tikztotarget)!1.5em!-90:(\tikztostart)$)
        in(\p1) -- (\p2) node[pos=.5,sloped,above]{#1};
        \endpgfinterruptpath
      }
    }
  },
  dim double above/.style={to path={\pgfextra{
        \pgfinterruptpath
        \draw[>=latex,|->|] let
        \p1=($(\tikztostart)!3em!90:(\tikztotarget)$),
        \p2=($(\tikztotarget)!3em!-90:(\tikztostart)$)
        in(\p1) -- (\p2) node[pos=.5,sloped,above]{#1};
        \endpgfinterruptpath
      }
    }
  },
  dim below/.style={to path={\pgfextra{
        \pgfinterruptpath
        \draw[>=latex,|->|] let
        \p1=($(\tikztostart)!-1em!-90:(\tikztotarget)$),
        \p2=($(\tikztotarget)!-1em!90:(\tikztostart)$)
        in (\p1) -- (\p2) node[pos=.5,sloped,below]{#1};
        \endpgfinterruptpath
      }
    }
  },
}

\tikzset{
    right angle quadrant/.code={
        \pgfmathsetmacro\quadranta{{1,1,-1,-1}[#1-1]}     % Arrays for selecting quadrant
        \pgfmathsetmacro\quadrantb{{1,-1,-1,1}[#1-1]}},
    right angle quadrant=1, % Make sure it is set, even if not called explicitly
    right angle length/.code={\def\rightanglelength{#1}},   % Length of symbol
    right angle length=2ex, % Make sure it is set...
    right angle symbol/.style n args={3}{
        insert path={
            let \p0 = ($(#1)!(#3)!(#2)$) in     % Intersection
                let \p1 = ($(\p0)!\quadranta*\rightanglelength!(#3)$), % Point on base line
                \p2 = ($(\p0)!\quadrantb*\rightanglelength!(#2)$) in % Point on perpendicular line
                let \p3 = ($(\p1)+(\p2)-(\p0)$) in  % Corner point of symbol
            (\p1) -- (\p3) -- (\p2)
        }
    }
}

\newcommand{\pgfextractangle}[3]{%
    \pgfmathanglebetweenpoints{\pgfpointanchor{#2}{center}}
                              {\pgfpointanchor{#3}{center}}
    \global\let#1\pgfmathresult
}

\usetikzlibrary{shapes.arrows}
\newcommand{\commentarrow}[1][Azure4]{\tikz[baseline=-3pt]{\node[shape border uses incircle, fill=#1,rotate=180,single arrow, inner sep=1pt, minimum size=6pt, single arrow head extend=2pt]{};}}

%% file: preamble/robotics.tex
\tikzset{ax/.style={-latex,line width=2pt}}

\tikzset{camera/.style={fill=Sienna1,fill opacity=0.5},%
image plane/.style={draw=RoyalBlue3,line width=2pt}}

%% file: preamble/reviewDraftClean.tex
\makeatletter
\newcommand{\@LN@col}[1]{}
\newcommand{\@LN}[2]{}
\makeatother